\documentclass[a4paper,10pt,fleqn]{article}
\usepackage[T1]{fontenc}
\usepackage{graphicx}
\usepackage{amsmath,amsfonts,amssymb,mathtools,bm}
\usepackage{amsthm}
\usepackage{booktabs,makecell,multirow,array,colortbl,dcolumn}
\usepackage{threeparttable,tabularx,longtable,adjustbox}
\usepackage{subcaption,float,placeins,pdflscape}
\usepackage{enumitem,microtype}
\usepackage[numbers,sort&compress]{natbib}
\usepackage[svgnames,dvipsnames]{xcolor}
\usepackage[colorlinks=true,linkcolor=DarkSlateGrey,urlcolor=DarkSlateGrey,citecolor=DarkSlateGrey]{hyperref}
\usepackage[nameinlink,noabbrev]{cleveref}
\usepackage{pgffor}
\usepackage[margin=1in]{geometry}
\graphicspath{{Figures/}{Supplementary_Figures/}}
\theoremstyle{plain}
\newtheorem{theorem}{Theorem}[section]

\theoremstyle{definition}
\newtheorem{assumption}{Assumption}[section]

\crefname{assumption}{Assumption}{Assumptions}
\Crefname{assumption}{Assumption}{Assumptions}
\newcommand{\R}{\mathbb{R}}
\newcommand{\E}{\mathbb{E}}

\newcommand{\norm}[1]{\left\lVert#1\right\rVert}
\newcommand{\abs}[1]{\left\lvert#1\right\rvert}
\newcommand{\ind}{\mathbb{I}}
\newcommand{\trans}{^{\mathsf T}}

\newcommand{\med}{\mathop{\rm Med}}

\newenvironment{keywords}{\par\noindent\textbf{Keywords:} }{\par}
\newcommand{\sep}{\unskip;\ }
\begin{document}
\title{Calibration of Time-Varying SEIR Models}

\author{
\parbox{0.90\textwidth}{
\centering
Yisa Abolade\textsuperscript{1} \quad
Yichuan Zhao\textsuperscript{1} \quad
Gerardo Chowell\textsuperscript{2}\\[0.5em]
{\small
\textsuperscript{1}Department of Mathematics and Statistics,
Georgia State University, Atlanta, Georgia, USA\\
\textsuperscript{2}Department of Population Health Sciences,
School of Public Health, Georgia State University,
Atlanta, Georgia, USA\\[0.5em]
\texttt{yabolade1@gsu.edu} \quad
\texttt{yichuan@gsu.edu} \quad
\texttt{gchowell@gsu.edu}
}
}
}

\date{}

\maketitle
\begin{abstract}
Mechanistic epidemic calibration can be sensitive to short-lived reporting anomalies that give a few observations high leverage under least squares (LSQ). We develop direct forward-solver least absolute deviations (LAD) calibration for time-varying susceptible--exposed--infectious--removed models fitted to reporting-interval incidence, with finite-horizon regularity and large-sample results. A phase-aware Monte Carlo study compares LAD and LSQ across three transmission drivers, four epidemic phases, five reporting mechanisms, and four forecast horizons. All 120,000 primary fits converged. LAD was favored in 56.7\% of mean-absolute-error comparisons and 57.9\% of weighted-interval-score comparisons overall; under isolated spikes, backlog release, or temporary underreporting, these proportions rose to 73.6\% and 72.2\%, with strong gains near the epidemic peak. Rolling-origin validation on four synthetic RAPIDD Ebola scenarios also identifies settings where LSQ is preferred, demonstrating the framework's ability to distinguish robustness gains from sustained-trend behavior. Multiple uncertainty and identifiability diagnostics, complete code and data, and an interactive application provide a reproducible computational workflow.
\end{abstract}

\section{Introduction}
\label{sec:introduction}
Ordinary differential equation models are used throughout medical research to translate biological mechanisms into longitudinal predictions. In infectious-disease surveillance, mechanistic epidemic models connect assumptions about infection, latency, removal, and intervention to observed case counts. Their inferential and forecasting value depends not only on the state equations, but also on the observation target, calibration criterion, transmission parameterization, numerical stability, and forecast evaluation design. These issues were prominent during the 2014--2015 West African Ebola epidemic, when delayed confirmation, incomplete reporting, behavioral adaptation, changing treatment capacity, and rapidly evolving control measures complicated real-time estimation and prediction \citep{chretien2015modeling,chowell2017perspectives}. More generally, epidemic forecasts are sensitive to aggregation, model structure, data quality, and departures from simple exponential growth \citep{chowell2016review,chowell2017primer}.

Least squares (LSQ) is a common calibration criterion for ordinary differential equation models and is well aligned with conditional-mean estimation and squared-error evaluation under approximately homoscedastic, light-tailed errors. Its quadratic loss gives large residuals substantial leverage, which is important in surveillance series containing isolated spikes, backlog releases, temporary underreporting with delayed release, duplicate removal, or serially correlated reporting error. Least absolute deviations (LAD), equivalently minimizing the sum of absolute deviations (SAD), increases linearly with residual magnitude and targets conditional-median behavior under the working additive-error model \citep{huber2009robust,koenker2005quantile}. This complementary estimand makes LAD especially attractive as a prespecified robustness analysis when short-lived reporting anomalies may otherwise dominate calibration.

Robust estimation for nonlinear ODE models has been developed through robust generalized profiling \citep{cao2011robust}, and Huber M-estimation with established asymptotic properties has also been studied \citep{qiu2016robust}. Building on this foundation, we develop a direct forward-solver LAD framework for sparsely observed reporting-interval incidence and evaluate it under matched predictive uncertainty. The central statistical questions are when LAD improves out-of-sample prediction, how the gain varies with epidemic phase and reporting mechanism, and how robust calibration interacts with transmission structure and forecast horizon. The contribution is a phase-aware, computationally reproducible framework that connects robust ODE calibration, predictive scoring, and practical identifiability.

Because LAD and LSQ target different estimands, rigorous comparison benefits from varying the forecast origin relative to epidemic phase and evaluating multiple forecast horizons. We therefore combine point scores with proper probabilistic scores, coverage, and interval width, and place mechanistic models alongside simple forecasting baselines. This design distinguishes robust-calibration gains from improvements attributable to transmission structure or favorable in-sample fit and yields a stringent out-of-sample assessment.

Time-varying transmission creates a second methodological opportunity. Treatment-unit capacity, case isolation, contact tracing, safe burial, risk perception, mobility, and changing contact behavior can all alter transmission. Flexible transmission functions can improve trajectory representation, while explicit parameterization and sensitivity diagnostics help preserve interpretable mechanistic structure. A representation $\beta(t)=\beta_0\{1+\alpha g(t;\vartheta)\}$ is structurally redundant when $g$ contains another unrestricted amplitude, because only the product of the amplitudes is identified. We use one amplitude parameter per driver, impose positive transmission over the full calibration and forecast horizon, and assess practical identifiability with sensitivity singular values, condition numbers, and profile objectives \citep{chowell2017primer,roosa2019identifiability}.

The empirical application uses four synthetic outbreaks released for the RAPIDD Ebola Forecasting Challenge. The challenge was generated from a spatially structured individual-based model representing transmission across Liberia and was designed to reproduce both epidemiological variability and the reporting uncertainty experienced during the West African epidemic \citep{ajelli2018rapiddmodel}. The scenarios differed in intervention trajectories, information availability, stochastic evolution, underreporting, and reporting reliability. Forecasting performance across the challenge depended strongly on the amount and quality of information and on forecast horizon \citep{viboud2018rapidd,pell2018phenomenological}.

We characterize the epidemic phases and reporting mechanisms in which changing from squared to absolute residual weighting delivers meaningful gains in SEIR calibration and short-horizon forecasting.

This study makes three integrated methodological contributions. First, we formulate forward-solver LAD calibration for ODE-implied reporting-interval incidence and state the regularity, consistency, and asymptotic-normality results needed to interpret it. Second, an ADEMP-structured, phase-aware simulation separates the roles of residual loss, epidemic phase, reporting mechanism, transmission structure, and forecast horizon. Third, rolling-origin RAPIDD validation evaluates mechanistic forecasts against simple baselines using point and proper probabilistic scores, multiple uncertainty procedures, and identifiability diagnostics. Together, these components make calibration loss an empirically testable part of the statistical analysis. The complete implementation is released as an installable R package with reproducible analysis code, automated validation, and an interactive application.

\section{Epidemic model and observation process}
\label{sec:model}

\subsection{Reporting intervals and interval incidence}
Let $0=\tau_0<\tau_1<\cdots<\tau_T$ be reporting times, and let $Y_j$ denote cases reported during $(\tau_{j-1},\tau_j]$. We use the working observation model
\begin{equation}
Y_j=\mu_j(\theta_0)+\varepsilon_j,\qquad j=1,\ldots,T,
\label{eq:obsmodel}
\end{equation}
where $\mu_j(\theta)$ is model-implied interval incidence. For LAD theory, $\med(\varepsilon_j)=0$ is the defining location condition. The simulation and predictive analyses additionally use count distributions on the nonnegative integer scale.

Let $S(t)$, $E(t)$, $I(t)$, and $R(t)$ denote susceptible, exposed, infectious, and removed individuals in a closed population of size $N$. Introduce $C(t)$ as cumulative progression from exposed to infectious. The augmented SEIR model is
\begin{align}
\dot S(t)&=-\beta(t;\theta)\frac{S(t)I(t)}{N}, \label{eq:S}\\
\dot E(t)&=\beta(t;\theta)\frac{S(t)I(t)}{N}-\sigma E(t), \label{eq:E}\\
\dot I(t)&=\sigma E(t)-\gamma I(t), \label{eq:I}\\
\dot R(t)&=\gamma I(t), \label{eq:R}\\
\dot C(t)&=\sigma E(t),\qquad C(0)=0. \label{eq:C}
\end{align}
Here $1/\sigma$ is the mean latent period and $1/\gamma$ is the mean infectious period. The observation map is
\begin{equation}
\mu_j(\theta)=C(\tau_j;\theta)-C(\tau_{j-1};\theta)
=\int_{\tau_{j-1}}^{\tau_j}\sigma E(u;\theta)\,du.
\label{eq:interval-incidence}
\end{equation}
Thus the fitted quantity is reporting-interval incidence. Daily simulations use one-day intervals; the RAPIDD analysis uses seven-day intervals while retaining daily ODE rates.

\subsection{Time-varying transmission families}
Three low-dimensional families are considered:
\begin{align}
\beta_{\mathrm{cos}}(t;\theta)&=\beta_0\{1+a\cos(\omega t)\}, \label{eq:beta-cos}\\
\beta_{\mathrm{exp}}(t;\theta)&=\beta_0\{1+a\exp(bt)\}, \label{eq:beta-exp}\\
\beta_{\mathrm{log}}(t;\theta)&=\beta_0\left[q+\frac{1-q}{1+\exp\{k(t-\tau)\}}\right]. \label{eq:beta-log}
\end{align}
The cosine family is a parsimonious phenomenological turning-point curve over the finite analysis window. The exponential family represents monotone amplification. The logistic-decline family represents gradual strengthening of control, with long-run transmission fraction $q\in(0,1)$, decline rate $k>0$, and midpoint $\tau$.

The parameter vectors are
\[
\theta_{\mathrm{cos}}=(\beta_0,\sigma,\gamma,a,\omega,E_0,I_0)\trans,
\quad
\theta_{\mathrm{exp}}=(\beta_0,\sigma,\gamma,a,b,E_0,I_0)\trans,
\]
\[
\theta_{\mathrm{log}}=(\beta_0,\sigma,\gamma,q,k,\tau,E_0,I_0)\trans,
\]
with $R_0$ fixed and $S_0=N-E_0-I_0-R_0$. Each family has a single scale parameter $\beta_0$ and no duplicated amplitude.

\subsection{Feasible set and sampled-output identifiability}
Let $\mathcal H\subset\R^p$ be compact. The feasible set requires positive $\beta_0$, $\sigma$, and $\gamma$, nonnegative initial states, $E_0+I_0+R_0<N$, and
\begin{equation}
0<\underline\beta\le \beta(t;\theta)\le\overline\beta<\infty,
\qquad 0\le t\le \tau_T+h_{\max}.
\label{eq:beta-positive}
\end{equation}
For the cosine family, $|a|<1$ is sufficient for positivity. Positivity and the configured upper bound are checked numerically over the finite analysis horizon for all families.

Define $M_T(\theta)=(\mu_1(\theta),\ldots,\mu_T(\theta))\trans$ and $J_T(\theta)=\partial M_T(\theta)/\partial\theta\trans$. A parameter is locally sampled-output identifiable if equality of $M_T$ in a neighborhood implies equality of the parameter. Full column rank of $J_T(\theta_0)$ is a standard local sufficient condition, but numerical rank alone is inadequate when singular values are highly dispersed. We therefore report scaled singular values, effective numerical rank, and condition number, and we use profile objectives and fixed-rate sensitivity analyses to evaluate practical identifiability.

\section{Estimation, computation, and forecast evaluation}
\label{sec:estimation}

\subsection{Least absolute deviations and least squares}
Let $r_j(\theta)=Y_j-\mu_j(\theta)$. The two calibration criteria are
\begin{align}
L_T^{\mathrm{LAD}}(\theta)&=\sum_{j=1}^T\abs{r_j(\theta)}, \label{eq:lad}\\
L_T^{\mathrm{LSQ}}(\theta)&=\sum_{j=1}^T r_j(\theta)^2. \label{eq:lsq}
\end{align}
LAD is the standard name used throughout; it is equivalent to minimizing the sum of absolute deviations. Under the working location model, LAD targets conditional-median behavior and LSQ targets conditional-mean behavior. Their expected advantages therefore depend on the score. The evaluation metrics used in the simulation and RAPIDD analyses, together with their statistical targets and expected alignment with LAD or LSQ, are summarized in \Cref{tab:metric-targets}.

\begin{table}[t]
\centering
\caption{Evaluation metrics, statistical targets, and interpretation.}
\label{tab:metric-targets}
\small
\begin{tabularx}{\textwidth}{@{}lXXX@{}}
\toprule
Metric & Statistical target & Interpretation & Expected alignment\\
\midrule
MSE & Conditional mean / squared-error accuracy & Penalizes large point errors quadratically & LSQ is aligned with this criterion\\
MAE & Conditional median / absolute-error accuracy & Penalizes point errors linearly & LAD is aligned with this criterion\\
WIS & Full predictive distribution & Balances interval calibration and sharpness across levels & Depends on point fit and uncertainty method\\
Coverage & Calibration of predictive intervals & Fraction of observations contained in the stated interval & Depends on uncertainty calibration\\
Interval width & Sharpness of predictive intervals & Average width of the stated interval & Narrower is better only when coverage remains adequate\\
\bottomrule
\end{tabularx}
\end{table}

\subsection{Constrained multistart computation}
The state equations are integrated with LSODA \citep{hairer1996ode,soetaert2010desolve}, and interval incidence is obtained by differencing $C$ at adjacent reporting times. Bounded optimization uses L-BFGS-B \citep{byrd1995lbfgsb}. Within every paired simulation comparison, LAD and LSQ receive the same feasible starting values. The first start is biologically motivated and clipped to the configured bounds; additional starts are generated independently over the bounds, with positive parameters sampled on a logarithmic scale. Candidate vectors that imply invalid initial states, numerical integration failure, nonfinite trajectories, or transmission outside the admissible range are excluded. The primary fits use 12 starts and at most 3,000 optimizer iterations.

\subsection{Predictive distributions}
For the simulation study, each fitted trajectory is converted to a predictive distribution using a negative-binomial observation model. Dispersion is estimated from the calibration residuals, and 500 conditional predictive draws are generated for each fitted model and target time. This construction places LAD and LSQ forecasts on the same count scale while allowing variance to increase with the predicted mean.

For rolling-origin RAPIDD validation, the same negative-binomial predictive construction is used for all SEIR specifications and both baselines, with 1,000 draws per origin. The naive baseline carries the last observed count forward. The recent-window exponential baseline fits
\[
\log(Y_j+0.5)=\alpha+\rho j+e_j
\]
over the latest six observations and extrapolates the fitted trend, with $\rho$ constrained to $[-0.75,0.75]$ for numerical stability. Both baselines use calibration residuals to estimate predictive dispersion.

The fixed-origin sensitivity analysis compares four refitted uncertainty procedures. The IID residual bootstrap resamples median-centered residuals; the wild bootstrap multiplies residuals by independent random signs; the moving-block bootstrap samples contiguous residual blocks of length four; and the negative-binomial parametric bootstrap generates counts around the fitted mean with estimated dispersion. Each bootstrap data set is refitted from a warm start plus two additional feasible starts. The IID and negative-binomial procedures use 2,000 replications per model; the wild and moving-block procedures use 1,000. Percentile parameter intervals use successful refits only \citep{efron1993bootstrap,wu1986wild,kunsch1989jackknife}.

\subsection{Forecast scores and diagnostics}
For a central $(1-\alpha)$ prediction interval $[L_{j,\alpha},U_{j,\alpha}]$, the interval score is
\begin{equation}
\mathrm{IS}_{\alpha}(Y_j)=(U_{j,\alpha}-L_{j,\alpha})+
\frac{2}{\alpha}(L_{j,\alpha}-Y_j)_++
\frac{2}{\alpha}(Y_j-U_{j,\alpha})_+.
\label{eq:intervalscore}
\end{equation}
Given central intervals indexed by $\alpha_1,\ldots,\alpha_K$ and predictive median $m_j$, the weighted interval score is
\begin{equation}
\mathrm{WIS}(Y_j)=\frac{1}{K+1/2}\left\{
\frac12\abs{Y_j-m_j}+\sum_{k=1}^K\frac{\alpha_k}{2}\mathrm{IS}_{\alpha_k}(Y_j)
\right\}.
\label{eq:wis}
\end{equation}
WIS is proper and balances sharpness with penalties for observations outside the predictive intervals \citep{gneiting2007proper,bracher2021wis}. We report MSE, MAE, WIS, empirical coverage, and mean interval width at the 50\%, 80\%, 90\%, and 95\% levels. Randomized probability integral transform (PIT) histograms provide a distributional diagnostic for discrete forecasts \citep{czado2009predictive}.

\section{Large-sample theory for LAD calibration}
\label{sec:theory}
Let $x=(S,E,I,R)\trans$ and $\Delta_N=\{x\in\R_+^4:\mathbf 1\trans x=N\}$. Complete proofs are in the supplementary material.

\begin{theorem}[Finite-horizon global well-posedness]
\label{thm:wellposed}
Suppose $x(0)\in\Delta_N$, $\sigma>0$, $\gamma>0$, and $t\mapsto\beta(t;\theta)$ is continuous and nonnegative on $[0,\tau_T+h_{\max}]$. Then \eqref{eq:S}--\eqref{eq:R} admit a unique solution on the entire interval. Moreover, $x(t)\in\Delta_N$ for every $t$, and $C(t)$ is nondecreasing and finite.
\end{theorem}

\begin{theorem}[Parameter differentiability]
\label{thm:smooth}
Assume that the vector field and initial-condition map are continuously differentiable in $\theta$ on a neighborhood of the compact feasible set. Then $x(t;\theta)$, $C(t;\theta)$, and every interval output $\mu_j(\theta)$ are continuously differentiable in $\theta$, uniformly on the finite analysis horizon. Their derivatives satisfy the standard sensitivity equations.
\end{theorem}

Define
\[
\bar L_T(\theta)=T^{-1}\sum_{j=1}^T\abs{Y_j-\mu_j(\theta)},
\qquad Q_T(\theta)=\E\{\bar L_T(\theta)\}.
\]

\begin{assumption}[Uniform regularity]
\label{ass:uniform}
The set $\mathcal H$ is compact; $\sup_{T,j,\theta\in\mathcal H}\abs{\mu_j(\theta)}<\infty$; and there exists $K<\infty$ such that
\[
\abs{\mu_j(\theta)-\mu_j(\vartheta)}\le K\norm{\theta-\vartheta}
\]
for every $T,j$ and $\theta,\vartheta\in\mathcal H$. The errors are independent, have median zero, and satisfy $\sup_{T,j}\E\varepsilon_{Tj}^2<\infty$.
\end{assumption}

\begin{assumption}[Uniform separation]
\label{ass:separation}
For every $\epsilon>0$,
\[
\liminf_{T\to\infty}\inf_{\theta\in\mathcal H:\norm{\theta-\theta_0}\ge\epsilon}
\{Q_T(\theta)-Q_T(\theta_0)\}>0.
\]
\end{assumption}

\begin{theorem}[Existence, uniform convergence, and consistency]
\label{thm:consistency}
If $\mu_j(\theta)$ is continuous on compact $\mathcal H$, the LAD and LSQ objectives attain their minima. Under \cref{ass:uniform},
\[
\sup_{\theta\in\mathcal H}\abs{\bar L_T(\theta)-Q_T(\theta)}\xrightarrow{p}0.
\]
Under \cref{ass:uniform,ass:separation}, every measurable LAD minimizer $\widehat\theta_T$ satisfies $\widehat\theta_T\xrightarrow{p}\theta_0$.
\end{theorem}

Let $J_{Tj}=\partial\mu_j(\theta_0)/\partial\theta$ and $A_T=T^{-1}\sum_{j=1}^T J_{Tj}J_{Tj}\trans$.

\begin{assumption}[Local asymptotic conditions]
\label{ass:an}
The true parameter is an interior point of $\mathcal H$. The errors are identically distributed with median zero and density $f$ continuous and positive at zero. Further, $A_T\to A$ with $A$ positive definite, $\max_j\norm{J_{Tj}}/\sqrt T\to0$, the interval-output map has a uniform local quadratic remainder, and a Lindeberg condition holds for the deterministic Jacobian array.
\end{assumption}

\begin{theorem}[Asymptotic normality of LAD]
\label{thm:an}
Under \cref{ass:uniform,ass:separation,ass:an},
\[
\sqrt T(\widehat\theta_T-\theta_0)\xrightarrow{d}
N\left(0,\frac{1}{4f(0)^2}A^{-1}\right).
\]
\end{theorem}

These results provide a theoretical foundation for LAD calibration under increasing-reporting designs with interior parameters and stable local identification. The RAPIDD application complements the asymptotic analysis with finite-sample rolling-origin validation over 15--35 calibration observations and with explicit sensitivity diagnostics, so forecast performance and parameter identifiability are assessed on their appropriate scales.

\section{Phase-aware Monte Carlo study}
\label{sec:simulation}

\subsection{Design}
The simulation follows the ADEMP structure for evaluating statistical methods \citep{morris2019simulation}. Daily interval incidence is generated from the three transmission families in \cref{eq:beta-cos,eq:beta-exp,eq:beta-log} with $N=1{,}000{,}000$, $(E_0,I_0,R_0)=(40,20,0)$, $\beta_0=0.36$ day$^{-1}$, $\sigma=1/7$ day$^{-1}$, and $\gamma=1/6.5$ day$^{-1}$. The driver-specific values are $a=0.25$ and $\omega=2\pi/120$ for the cosine curve, $a=0.60$ and $b=-0.015$ for the exponential curve, and $(q,k,\tau)=(0.25,0.08,70)$ for logistic decline. Each trajectory contains 240 daily reporting intervals. Baseline observations are negative-binomial with mean equal to model incidence and size parameter 30.

Forecast origins are selected from the deterministic mean curve. Early growth is the first pre-peak time at which incidence reaches 20\% of the peak; near peak is two intervals before the peak; early decline is the first post-peak time below 70\% of the peak; and late decline is the first post-peak time below 20\%. The resulting origins are driver specific. Forecast horizons are 1, 3, 5, and 10 intervals.

Five reporting conditions are examined in the 28 intervals immediately preceding each origin. Clean data receive no additional perturbation. Isolated spikes add positive increments at 10\% of eligible reporting times. A backlog mechanism retains 25\% of the counts for four days and releases the accumulated balance at the origin. The temporary underreporting with delayed release mechanism retains 50\% of counts for ten days and releases 90\% of the missing total at the origin. A serial mechanism adds AR(1) reporting error with $\phi=0.75$ and innovation scale tied to 20\% of local median incidence. The combination of three drivers, four phases, and five reporting mechanisms yields 60 conditions. Each condition uses 1,000 paired replications, the same observations and starts for LAD and LSQ, and 500 predictive draws per fitted model. The core reproducibility settings for these simulations and the RAPIDD analyses are summarized in \Cref{tab:repro-main}, with complete bounds, starts, seeds, and software details reported in the supplement.

\begin{table}[t]
\centering
\caption{Core reproducibility settings. Complete parameter bounds, starting-value rules, seed definitions, and software versions are reported in the supplement.}
\label{tab:repro-main}
\small
\begin{tabularx}{\textwidth}{@{}lX@{}}
\toprule
Component & Setting\\
\midrule
Monte Carlo design & 60 driver--phase--reporting conditions; 1,000 paired replications per condition; 500 predictive draws.\\
Simulation horizons & 1, 3, 5, and 10 daily intervals from early growth, near peak, early decline, and late decline.\\
RAPIDD validation & Rolling origins 15--35 and horizons 1--5 weeks; fixed-origin example uses weeks 1--30 for calibration and 31--35 for scoring.\\
Optimization & 12 shared feasible starts for paired LAD/LSQ fits; L-BFGS-B, maximum 3,000 iterations; LSODA tolerances $10^{-9}$.\\
Uncertainty & Negative-binomial predictive distributions for rolling validation; four refitted-bootstrap procedures in the fixed-origin sensitivity analysis.\\
Software & R 4.5.0, \texttt{deSolve} 1.42, Linux x86-64; base seed 20260627.\\
\bottomrule
\end{tabularx}
\end{table}

\subsection{Phase-specific forecast comparison}
All 120,000 primary simulation fits (60 conditions $\times$ 1,000 replications $\times$ 2 losses) returned nominal optimizer convergence. The paired difference is defined as LAD minus LSQ; negative values favor LAD for error scores. A winner is assigned only when the paired 95\% Monte Carlo confidence interval excludes zero. The resulting phase-specific counts of LAD-favored, LSQ-favored, and inconclusive comparisons are reported in \Cref{tab:simulation-phase-counts}.

\begin{table}[t]
\centering
\caption{Numbers of phase-specific forecast comparisons in which the 95\% paired Monte Carlo confidence interval favored LAD, favored LSQ, or included zero. Each phase contains 60 comparisons per metric (three transmission drivers, five reporting mechanisms, and four horizons).}
\label{tab:simulation-phase-counts}
\small
\begin{tabular}{@{}lccc ccc@{}}
\toprule
&\multicolumn{3}{c}{MAE}&\multicolumn{3}{c}{WIS}\\
\cmidrule(lr){2-4}\cmidrule(lr){5-7}
Phase & LAD & LSQ & No clear & LAD & LSQ & No clear\\
\midrule
Early growth & 34 & 22 & 4 & 36 & 22 & 2\\
Near peak & 40 & 9 & 11 & 43 & 8 & 9\\
Early decline & 34 & 7 & 19 & 32 & 9 & 19\\
Late decline & 28 & 8 & 24 & 28 & 8 & 24\\
\midrule
Total & 136 & 46 & 58 & 139 & 47 & 54\\
\bottomrule
\end{tabular}
\end{table}

Across the full 240-comparison design, LAD was favored in 136 MAE comparisons (56.7\%) and 139 WIS comparisons (57.9\%). The robustness advantage was concentrated where short-lived reporting anomalies create high-leverage residuals: across isolated spikes, backlog release, and temporary underreporting with delayed release, LAD was favored in 106 of 144 MAE comparisons (73.6\%) and 104 of 144 WIS comparisons (72.2\%). Clean and serially correlated settings provide complementary reference regimes for understanding how the advantage changes when the anomaly structure is absent or temporally persistent.

Phase also mattered. LAD was most frequently favored near the peak, where its paired intervals favored it in 40 of 60 MAE comparisons (66.7\%) and 43 of 60 WIS comparisons (71.7\%). During early growth, LAD retained a clear advantage for backlog release and temporary underreporting with delayed release. To quantify the magnitude of the gains beyond statistical separation, we calculated the paired relative change $100(\mathrm{MAE}_{\mathrm{LAD}}-\mathrm{MAE}_{\mathrm{LSQ}})/\mathrm{MAE}_{\mathrm{LSQ}}$ for each matched replication. The median and interquartile range by reporting mechanism and phase are shown in \Cref{tab:relative-mae}. Relative differences are undefined when the matched LSQ MAE is zero; 416 of 240,000 matched forecast comparisons (0.17\%) were therefore omitted only from this normalized summary. The largest median LAD reductions occur for backlog release and temporary underreporting with delayed release during early growth and near the peak, identifying the regimes in which robust residual weighting has the greatest practical effect.

\begin{table}[p]
\centering
\small
\caption{Magnitude of paired LAD-versus-LSQ forecast differences. Entries are the median percentage change in MAE, $100(\mathrm{MAE}_{\mathrm{LAD}}-\mathrm{MAE}_{\mathrm{LSQ}})/\mathrm{MAE}_{\mathrm{LSQ}}$, with the interquartile range (IQR) in parentheses, pooled over drivers, horizons, and replications. Negative values favor LAD. Ratios are omitted when the matched LSQ MAE is zero (416 of 240,000 comparisons).}
\label{tab:relative-mae}
\begin{tabular}{llr}
\toprule
Reporting mechanism & Epidemic phase & Median \% change (IQR) \\
\midrule
Clean & Early growth & 6.1 (-40.5, 106.2) \\
 & Near peak & 4.2 (-42.2, 91.4) \\
 & Early decline & 0.0 (-42.7, 78.0) \\
 & Late decline & 0.0 (-53.2, 114.2) \\
\addlinespace
Isolated spikes & Early growth & 7.8 (-44.0, 126.0) \\
 & Near peak & -7.1 (-55.5, 79.0) \\
 & Early decline & -13.8 (-59.1, 70.5) \\
 & Late decline & -11.0 (-62.7, 71.4) \\
\addlinespace
Backlog release & Early growth & -72.9 (-89.1, -21.6) \\
 & Near peak & -45.0 (-79.0, 41.8) \\
 & Early decline & -15.4 (-60.6, 72.0) \\
 & Late decline & 9.2 (-50.0, 146.7) \\
\addlinespace
Temporary underreporting with delayed release & Early growth & -84.2 (-92.9, -65.4) \\
 & Near peak & -56.2 (-78.5, -10.2) \\
 & Early decline & -27.5 (-74.8, 67.1) \\
 & Late decline & 14.5 (-48.9, 160.7) \\
\addlinespace
Serially correlated reporting error & Early growth & 10.5 (-39.1, 114.1) \\
 & Near peak & 0.0 (-48.7, 79.0) \\
 & Early decline & -6.8 (-50.8, 62.4) \\
 & Late decline & 0.0 (-52.9, 107.7) \\
\addlinespace
\bottomrule
\end{tabular}
\end{table}

The full phase-aware winner pattern is displayed for MAE in \Cref{fig:sim-mae-winner} and for WIS in \Cref{fig:sim-wis-winner}.

\begin{figure}[p]
\centering
\includegraphics[width=\textwidth]{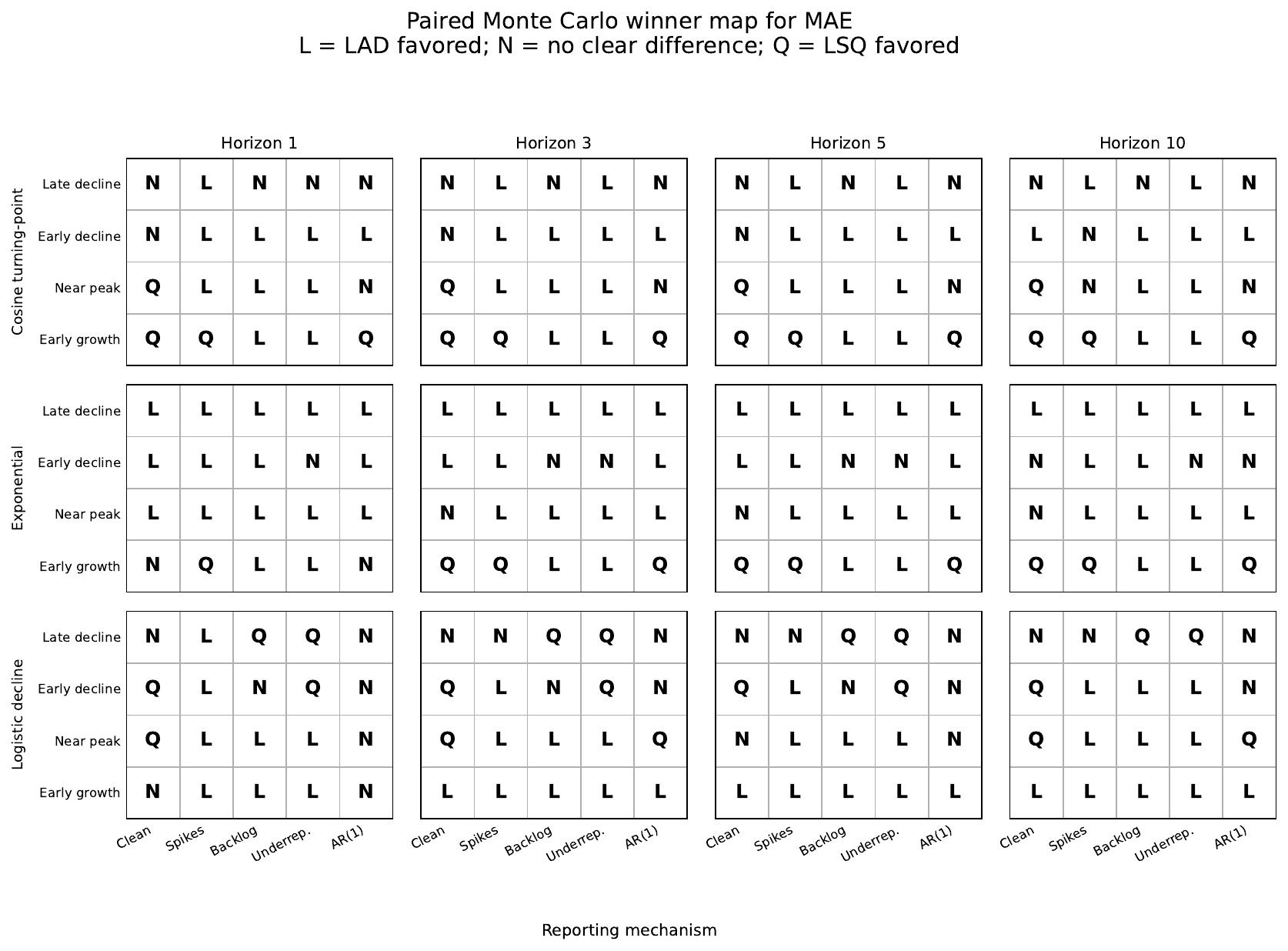}
\caption{Phase-aware paired Monte Carlo winner map for forecast mean absolute error (MAE). Cell symbols are L (LAD favored), N (no clear difference), and Q (LSQ favored) according to the 95\% paired Monte Carlo confidence interval. AR(1) denotes first-order autoregressive reporting error; Underrep. denotes temporary underreporting with delayed release.}
\label{fig:sim-mae-winner}
\end{figure}

\begin{figure}[p]
\centering
\includegraphics[width=\textwidth]{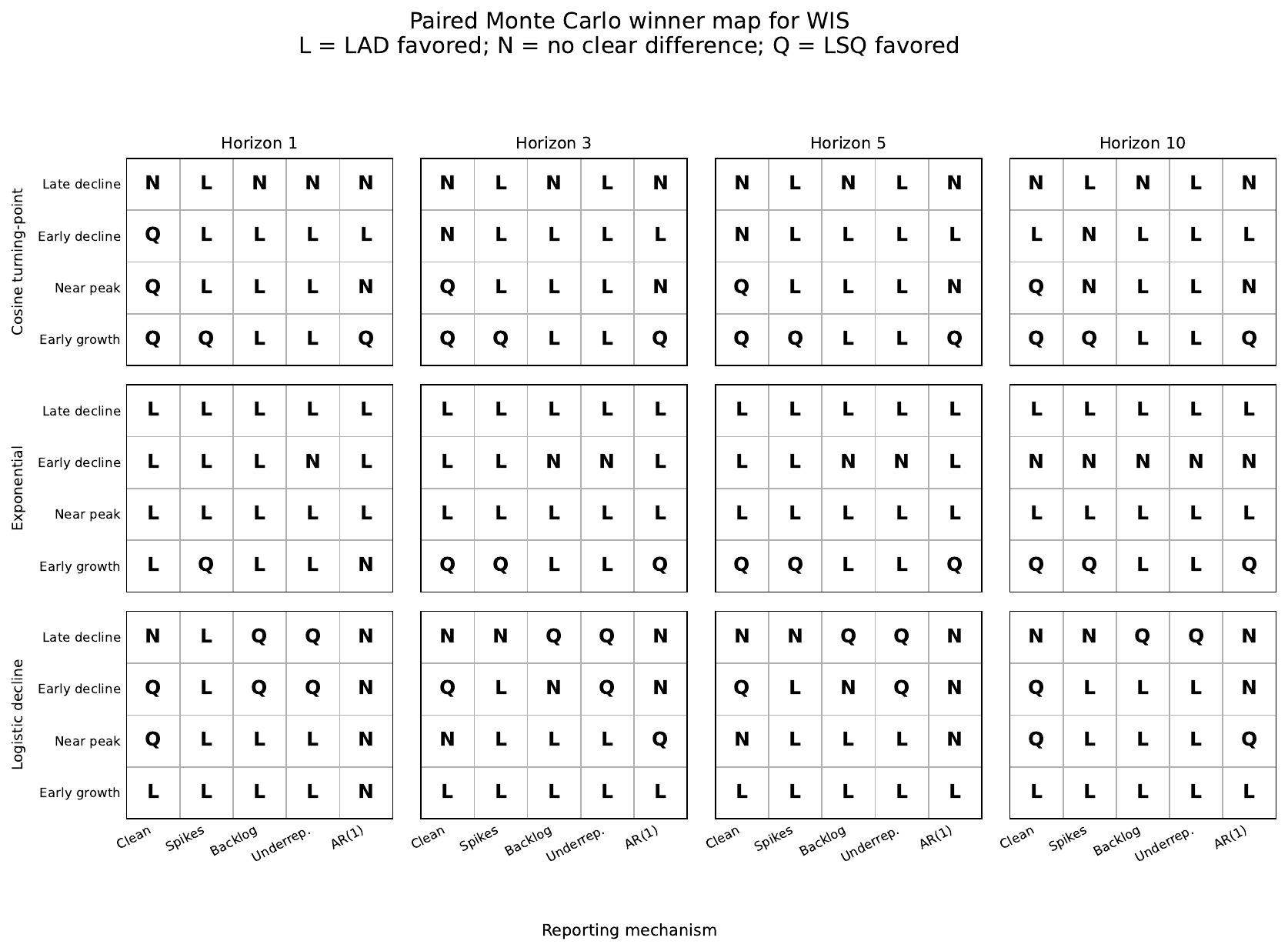}
\caption{Phase-aware paired Monte Carlo winner map for weighted interval score (WIS). Cell symbols are L (LAD favored), N (no clear difference), and Q (LSQ favored). AR(1) denotes first-order autoregressive reporting error; Underrep. denotes temporary underreporting with delayed release.}
\label{fig:sim-wis-winner}
\end{figure}

Parameter recovery provides additional evidence of robust-estimation gains. LAD had lower RMSE for $E_0$ in 56 of 60 conditions, for $\sigma$ in 48 of 60, for the cosine frequency $\omega$ in all 20 relevant conditions, and for logistic parameters $q$ and $k$ in 18 of 20 and 15 of 20 conditions, respectively. These results show that robust trajectory weighting can improve recovery of several important initial-state, progression, and transmission-shape components in addition to improving forecast scores.

\section{Synthetic Ebola forecasting application}
\label{sec:application}

\subsection{RAPIDD scenarios and released incidence series}
The application uses the four synthetic outbreaks released for the RAPIDD Ebola Forecasting Challenge \citep{ajelli2018rapiddmodel,viboud2018rapidd}. The challenge generator represented Ebola transmission across the 15 counties of Liberia and incorporated treatment-unit capacity, contact tracing, safe burial, case isolation, behavioral change, mobility, underreporting, delayed reporting, and incomplete case information. The scenarios are distinct stochastic realizations with different intervention trajectories and information environments; they are not four fixed intervention percentages and they should not be interpreted as four observed national outbreaks.

Each released series contains 46 weekly incidence observations from January 13 through November 24, 2014. The values are used directly, without interpolation, smoothing, differencing, or imputation. Scenario 1 is a comparatively data-rich outbreak controlled after progressive strengthening of intervention. Scenario 2 has lower information availability, is ultimately controlled, and peaks later. Scenario 3 combines a distinctive early stochastic path with an abrupt intervention change before decline. Scenario 4 contains prolonged growth followed by an abrupt drop in the released series. The largest Scenario 4 decrease occurs at week 40 (October 13, 2014), when the released count falls by 494 cases. This drop is outside the prespecified fixed-origin scoring window (weeks 31--35), but it is included in some later rolling-origin targets. We treat it as an abrupt structural reporting change in the released synthetic series rather than as evidence of instantaneous elimination of transmission. The four released scenario trajectories are shown in \Cref{fig:rapidd-overview}.

\begin{figure}[t]
\centering
\includegraphics[width=0.94\textwidth]{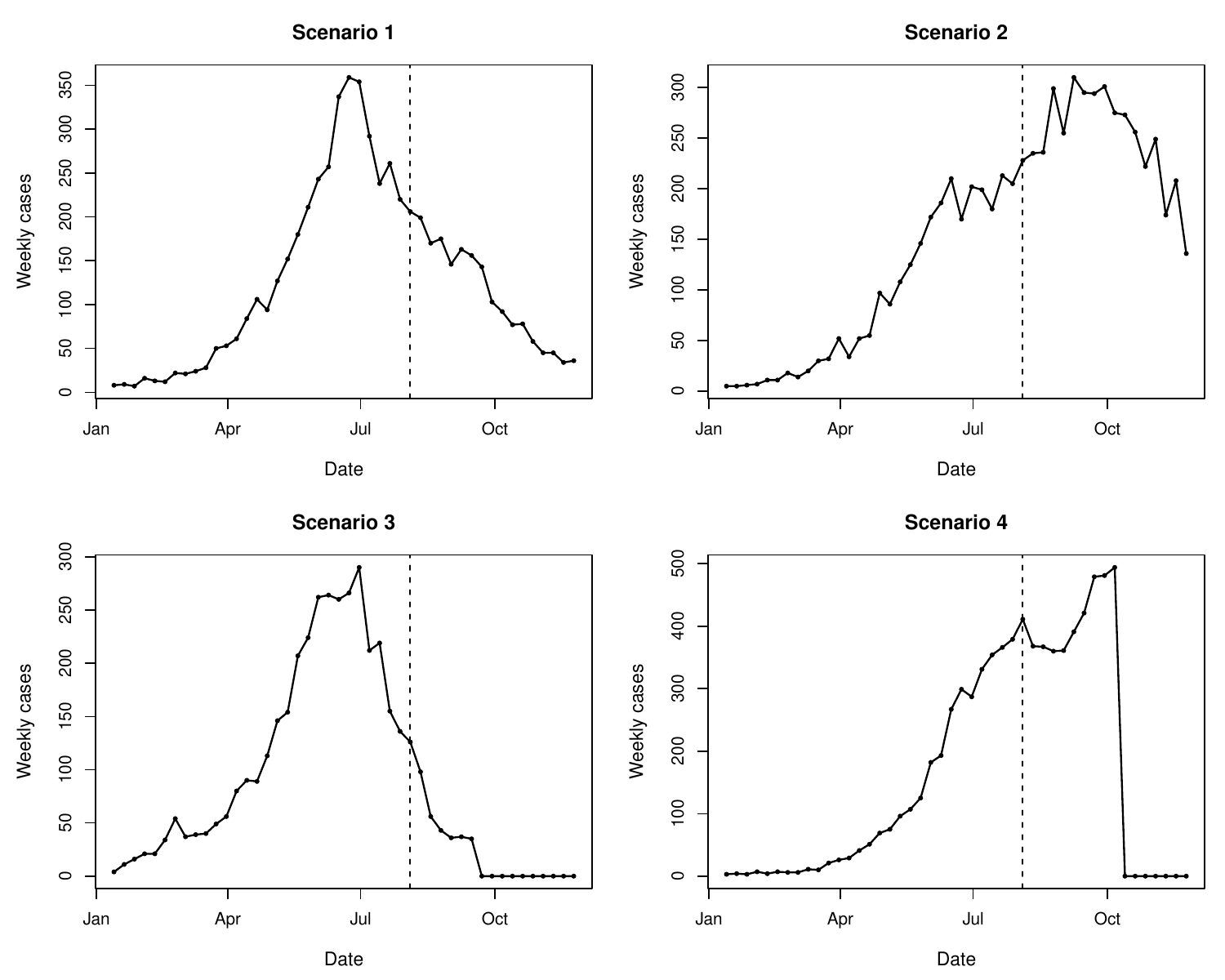}
\caption{Weekly released case counts in the four synthetic RAPIDD Ebola scenarios. The vertical scale is scenario-specific. The abrupt late drop in Scenario 4 occurs outside the prespecified weeks 31--35 fixed-origin scoring window, but inside some rolling-origin targets; it is treated as an abrupt structural reporting change in the released series.}
\label{fig:rapidd-overview}
\end{figure}

\subsection{Forecasting protocol}
ODE rates are expressed per day, with population fixed at $N=4{,}499{,}621$. The cosine turning-point, exponential, and logistic-decline transmission families are fitted under both LAD and LSQ. The primary validation is rolling origin: calibration endpoints are weeks 15 through 35 and horizons are 1 through 5 weeks. Every model is evaluated on the same available origin--horizon pairs. The two simple baselines described in \cref{sec:estimation} are included in the same evaluation.

The prespecified fixed-origin analysis is retained as a transparent example. Weeks 1--30 are used for calibration, weeks 31--35 are scored, and weeks 36--46 are not used in that fixed-origin score. They remain visible in the descriptive series and contribute to rolling-origin validation when they fall within a target horizon. Thus the late Scenario 4 drop is deliberately excluded from the fixed-origin score, not removed from the data set.

Rolling-origin predictive distributions use the negative-binomial construction for all eight candidate models: six SEIR combinations (three transmission drivers by two losses), the naive last-observation forecast, and the recent-window exponential forecast. The implementation does not hold one dispersion value fixed for an entire model--origin fit. Instead, within each predictive replication, residuals from that model's calibration window are sampled with replacement to the required forecast length. If $e_h^*$ denotes a sampled residual and $\widehat\mu_h$ the forecast mean, the pseudo-count is $y_h^*=\widehat\mu_h+e_h^*$ and the moment estimator is
\[
\widehat v^*=\frac{1}{H}\sum_{h=1}^{H}(y_h^*-\widetilde\mu_h)^2,\qquad
\overline\mu=\frac{1}{H}\sum_{h=1}^{H}\widetilde\mu_h,\qquad
\overline{\mu^2}=\frac{1}{H}\sum_{h=1}^{H}\widetilde\mu_h^2,
\]
where $\widetilde\mu_h=\max(\widehat\mu_h,10^{-6})$, followed by
\[
\widehat\kappa^*=\frac{\overline{\mu^2}}{\widehat v^*-\overline\mu}.
\]
Thus the size parameter is re-estimated within each predictive draw, using residuals from the current model--origin calibration fit. No degrees-of-freedom correction is applied. If $\widehat v^*-\overline\mu\le 10^{-8}$ or the estimate is nonfinite, the size is set to 30; otherwise it is truncated to $[0.1,10^6]$. Observed zero counts are retained; only model means are floored for numerical validity, and the negative-binomial generator uses a mean floor of $10^{-8}$. The same moment formula and constraints are used for SEIR and baseline forecasts, although each model supplies its own fitted values and residuals. Expanded notes in the supplement distinguish this rolling predictive construction from the fixed-origin negative-binomial bootstrap. The fixed-origin uncertainty sensitivity analysis compares IID residual, wild, moving-block, and negative-binomial parametric bootstraps. Scores include MAE, WIS, empirical coverage, and interval width at the 50\%, 80\%, 90\%, and 95\% levels. Randomized PIT values are also retained for diagnostic assessment.

\section{RAPIDD forecasting results}
\label{sec:rapidd-results}

\subsection{Rolling-origin validation}
The rolling-origin analysis contains 21 calibration endpoints per scenario, five horizons, and eight candidate models, subject only to the availability of future observations near the end of each series. The same origin--horizon pairs are used within every comparison. The resulting winners by scenario and forecast horizon are summarized in \Cref{tab:rapidd-winners}.

\begin{table}[t]
\centering
\caption{Rolling-origin forecast winners across calibration endpoints 15--35. Parentheses give the mean score over available origins. MAE denotes mean absolute error; WIS denotes weighted interval score; NL denotes naive last observation; RE denotes recent-window exponential; and LD-LSQ denotes logistic-decline transmission calibrated by least squares. ``Best SEIR'' is selected by WIS.}
\label{tab:rapidd-winners}
\scriptsize
\begin{tabular}{@{}llccc@{}}
\toprule
Scenario & Horizon & Lowest MAE & Lowest WIS & Best SEIR (WIS)\\
\midrule
Scenario 1 & 1 & NL (27.00) & NL (16.62) & LD-LSQ (24.44)\\
 & 2 & RE (44.24) & NL (28.66) & LD-LSQ (38.34)\\
 & 3 & NL (63.86) & NL (44.14) & LD-LSQ (52.96)\\
 & 4 & NL (81.24) & NL (58.56) & LD-LSQ (78.01)\\
 & 5 & NL (92.95) & NL (70.15) & LD-LSQ (102.3)\\
\addlinespace
Scenario 2 & 1 & LD-LSQ (18.79) & RE (12.51) & LD-LSQ (13.85)\\
 & 2 & LD-LSQ (25.05) & RE (16.57) & LD-LSQ (17.36)\\
 & 3 & NL (34.38) & NL (24.20) & LD-LSQ (25.65)\\
 & 4 & NL (46.10) & NL (32.36) & LD-LSQ (37.76)\\
 & 5 & NL (51.12) & NL (36.76) & LD-LSQ (53.82)\\
\addlinespace
Scenario 3 & 1 & NL (22.71) & NL (14.98) & LD-LSQ (18.66)\\
 & 2 & RE (40.74) & RE (25.91) & LD-LSQ (30.89)\\
 & 3 & NL (58.52) & NL (42.13) & LD-LSQ (48.04)\\
 & 4 & NL (76.12) & NL (57.05) & LD-LSQ (67.25)\\
 & 5 & NL (91.60) & NL (70.83) & LD-LSQ (87.01)\\
\addlinespace
Scenario 4 & 1 & NL (24.90) & NL (16.24) & LD-LSQ (18.29)\\
 & 2 & NL (44.95) & NL (29.30) & LD-LSQ (31.27)\\
 & 3 & NL (64.14) & NL (45.04) & LD-LSQ (49.81)\\
 & 4 & NL (82.81) & NL (61.70) & LD-LSQ (76.29)\\
 & 5 & NL (112.9) & NL (91.23) & LD-LSQ (113.9)\\
\bottomrule
\end{tabular}
\begin{minipage}{0.96\textwidth}\footnotesize NL: naive last observation; RE: recent-window exponential; C: cosine turning-point; E: exponential; LD: logistic decline. The simple baselines achieved the lowest WIS in all 20 scenario--horizon cells. LD-LSQ was the best SEIR model in every cell.\end{minipage}
\end{table}

The rolling-origin analysis provides a stringent external benchmark for the mechanistic fits. A simple baseline achieved the smallest mean WIS in each of the 20 scenario--horizon cells: the naive last-observation model in 18 cells and the recent-window exponential baseline for Scenario 2 at horizons 1 and 2. Point-error rankings were similar, while logistic-decline LSQ achieved the lowest mean MAE for Scenario 2 at horizons 1 and 2. These benchmark results make the mechanistic comparisons deliberately conservative and strengthen the interpretation of any robust-calibration gains.

The RAPIDD exercise demonstrates an important feature of the proposed framework: it distinguishes anomaly-driven robustness gains from forecasting regimes dominated by sustained local trends. In this setting, LAD functions as a diagnostic sensitivity analysis that reveals how conclusions depend on residual weighting, while the rolling-origin benchmark identifies the best-performing loss and transmission structure for the observed forecasting regime.

Within the mechanistic family, logistic-decline LSQ had the smallest rolling-origin WIS and MAE in all 20 scenario--horizon cells. For each transmission driver, LSQ also had lower mean rolling-origin MAE and WIS than LAD at every scenario--horizon combination. Together with the Monte Carlo results, this finding demonstrates the phase- and mechanism-specific nature of robust calibration: LAD is especially valuable for short reporting anomalies, whereas the RAPIDD rolling windows often contain sustained local trends that favor mean-oriented calibration. The logistic-decline structure is well aligned with gradual strengthening of control. Relative to the naive baseline, the WIS ratio for logistic-decline LSQ ranges from 0.90 to 1.47 across the 20 cells (\Cref{tab:wis-skill-ratios}), including a value below 1 for Scenario 2 at horizon 1.

\begin{table}[H]
\centering
\caption{Rolling-origin WIS skill ratios for the strongest SEIR specification, logistic-decline LSQ, relative to the naive last-observation baseline. A ratio below 1 favors the SEIR model; a ratio above 1 quantifies its forecasting penalty relative to the naive baseline.}
\label{tab:wis-skill-ratios}
\begin{tabular}{lccccc}
\toprule
Scenario & $h=1$ & $h=2$ & $h=3$ & $h=4$ & $h=5$ \\
\midrule
Scenario 1 & 1.47 & 1.34 & 1.20 & 1.33 & 1.46 \\
Scenario 2 & 0.90 & 1.03 & 1.06 & 1.17 & 1.46 \\
Scenario 3 & 1.25 & 1.12 & 1.14 & 1.18 & 1.23 \\
Scenario 4 & 1.13 & 1.07 & 1.11 & 1.24 & 1.25 \\
\bottomrule
\end{tabular}
\end{table}

The rolling-origin MAE and WIS curves underlying these comparisons are displayed in \Cref{fig:rolling-mae,fig:rolling-wis}.

\begin{figure}[p]
\centering
\includegraphics[width=0.97\textwidth]{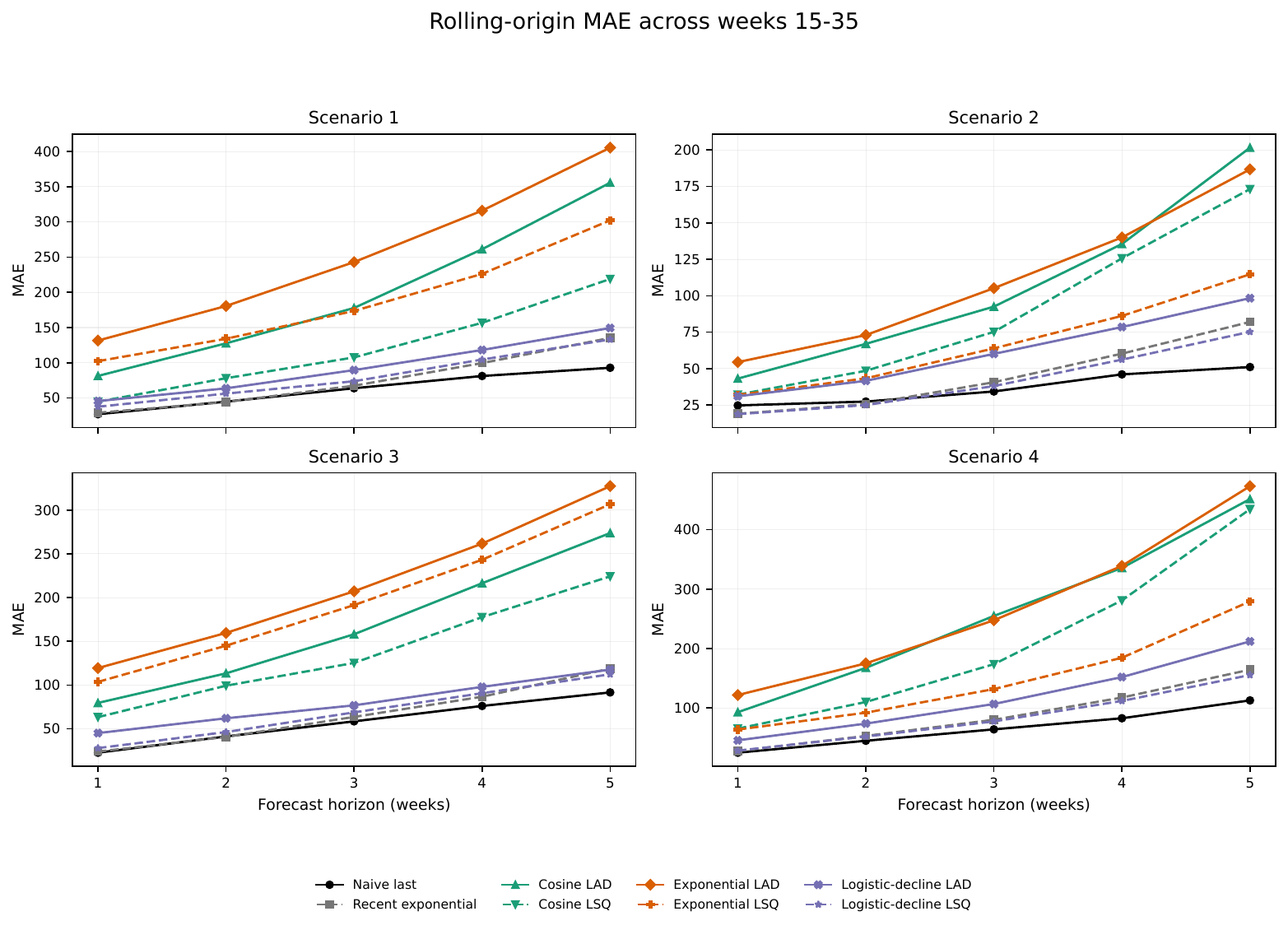}
\caption{Rolling-origin mean absolute error (MAE) by scenario and forecast horizon. Curves summarize origins 15--35. The naive and recent-window exponential baselines are included with the six SEIR specifications.}
\label{fig:rolling-mae}
\end{figure}

\begin{figure}[p]
\centering
\includegraphics[width=0.97\textwidth]{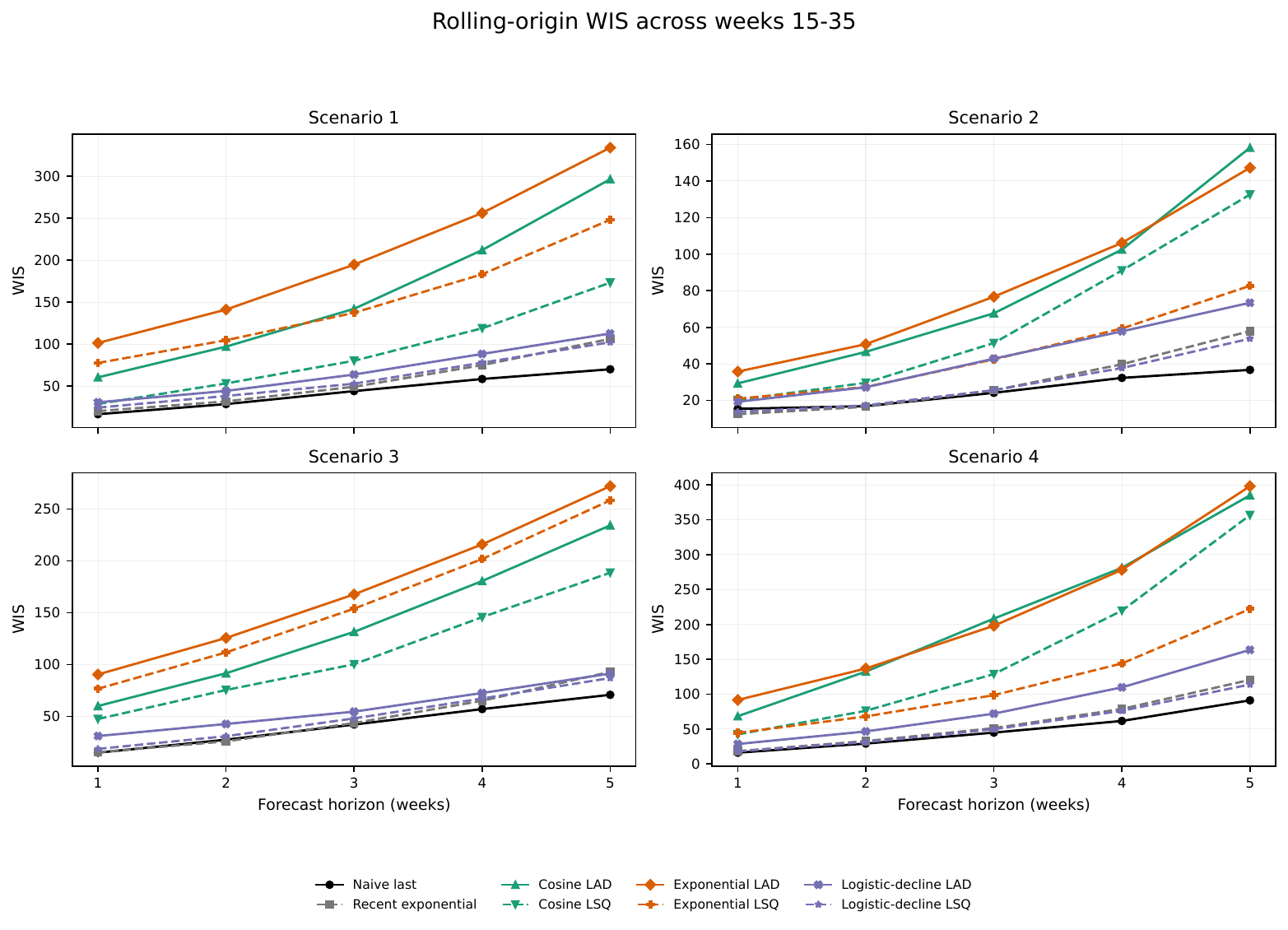}
\caption{Rolling-origin weighted interval score (WIS) by scenario and forecast horizon. Smaller values are preferred. Simple baselines achieved the smallest mean WIS in all 20 scenario--horizon cells, while logistic-decline LSQ was the best SEIR specification throughout.}
\label{fig:rolling-wis}
\end{figure}

\subsection{Prespecified fixed-origin example}
The fixed-origin analysis uses weeks 1--30 for calibration and weeks 31--35 for scoring. Its purpose is to show complete predictive trajectories under one prespecified origin and to support direct comparison of uncertainty procedures. It is not the sole basis for forecasting claims.

The best fixed-origin model varied by scenario and metric. Cosine LSQ had the lowest MAE and WIS in Scenario 1. Exponential LSQ had the lowest MAE and WIS in Scenario 2. The recent-window exponential baseline was best on both scores in Scenario 3. In Scenario 4, logistic-decline LSQ had the lowest MAE, while the naive last-observation forecast had the lowest WIS. Across matched transmission drivers, LSQ had lower fixed-origin MAE and WIS than LAD in all 12 driver--scenario comparisons. Complete values are reported in the supplement.

For clear visual comparison, the four fixed-origin scenario displays are presented separately with larger axes and a shared legend in \Cref{fig:fixed-s1,fig:fixed-s2,fig:fixed-s3,fig:fixed-s4}.

\begin{figure}[p]
\centering
\includegraphics[width=0.96\textwidth]{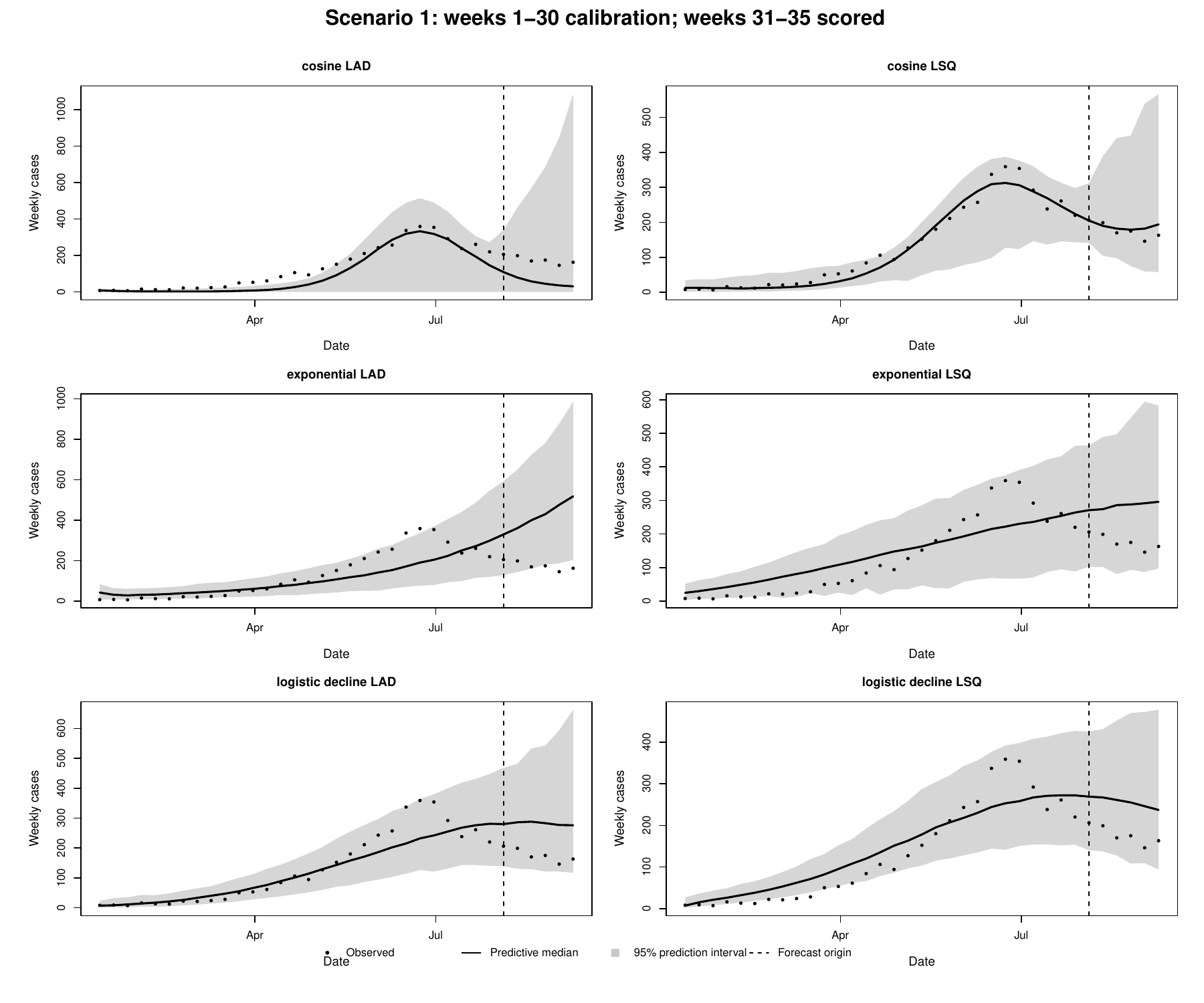}
\caption{Scenario 1: weeks 1--30 calibration and weeks 31--35 fixed-origin forecast comparison. Points are released weekly incidence; curves and intervals show the six fitted SEIR specifications. Baseline scores are reported in the accompanying tables.}
\label{fig:fixed-s1}
\end{figure}

\begin{figure}[p]
\centering
\includegraphics[width=0.96\textwidth]{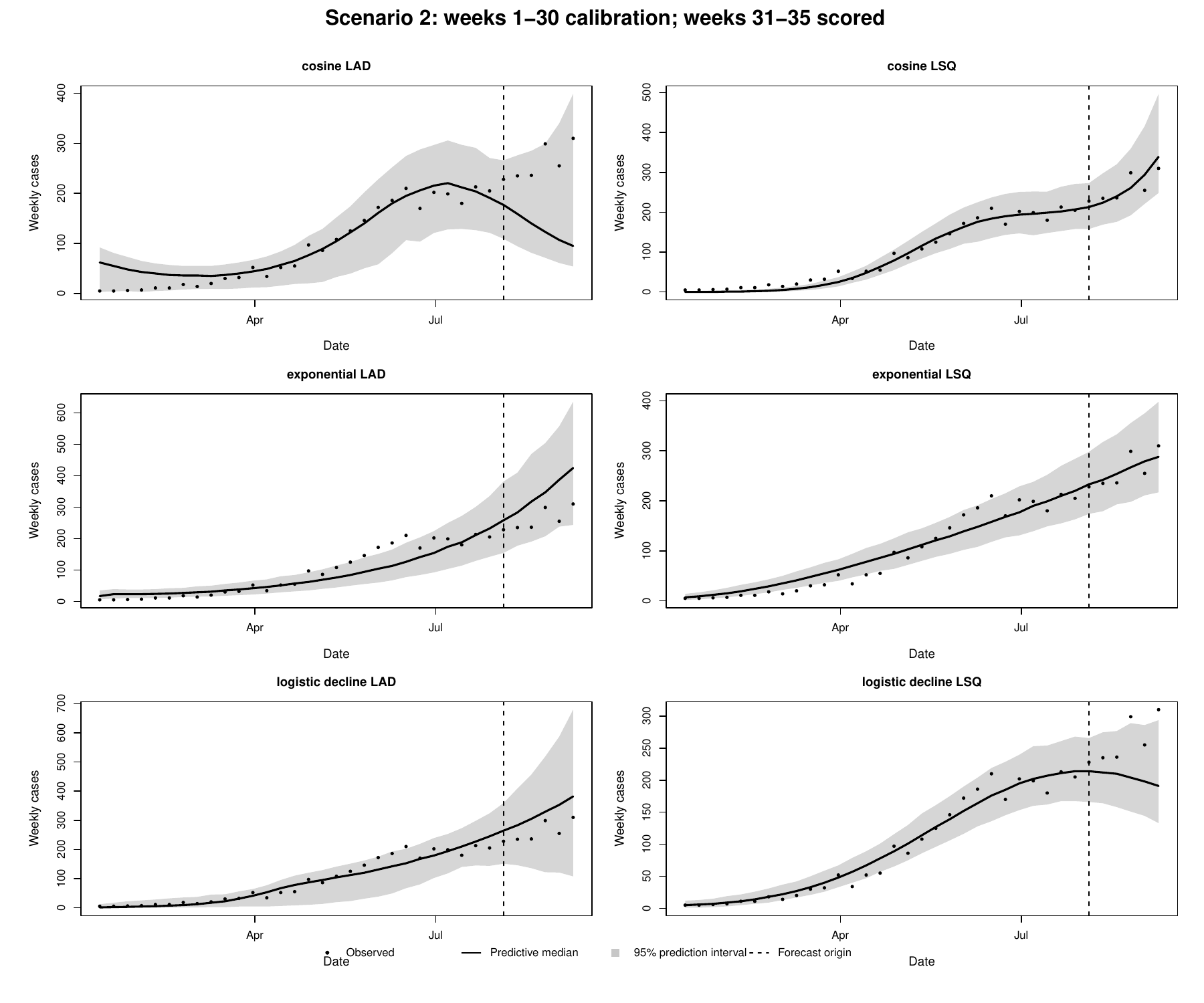}
\caption{Scenario 2: weeks 1--30 calibration and weeks 31--35 fixed-origin forecast comparison.}
\label{fig:fixed-s2}
\end{figure}

\begin{figure}[p]
\centering
\includegraphics[width=0.96\textwidth]{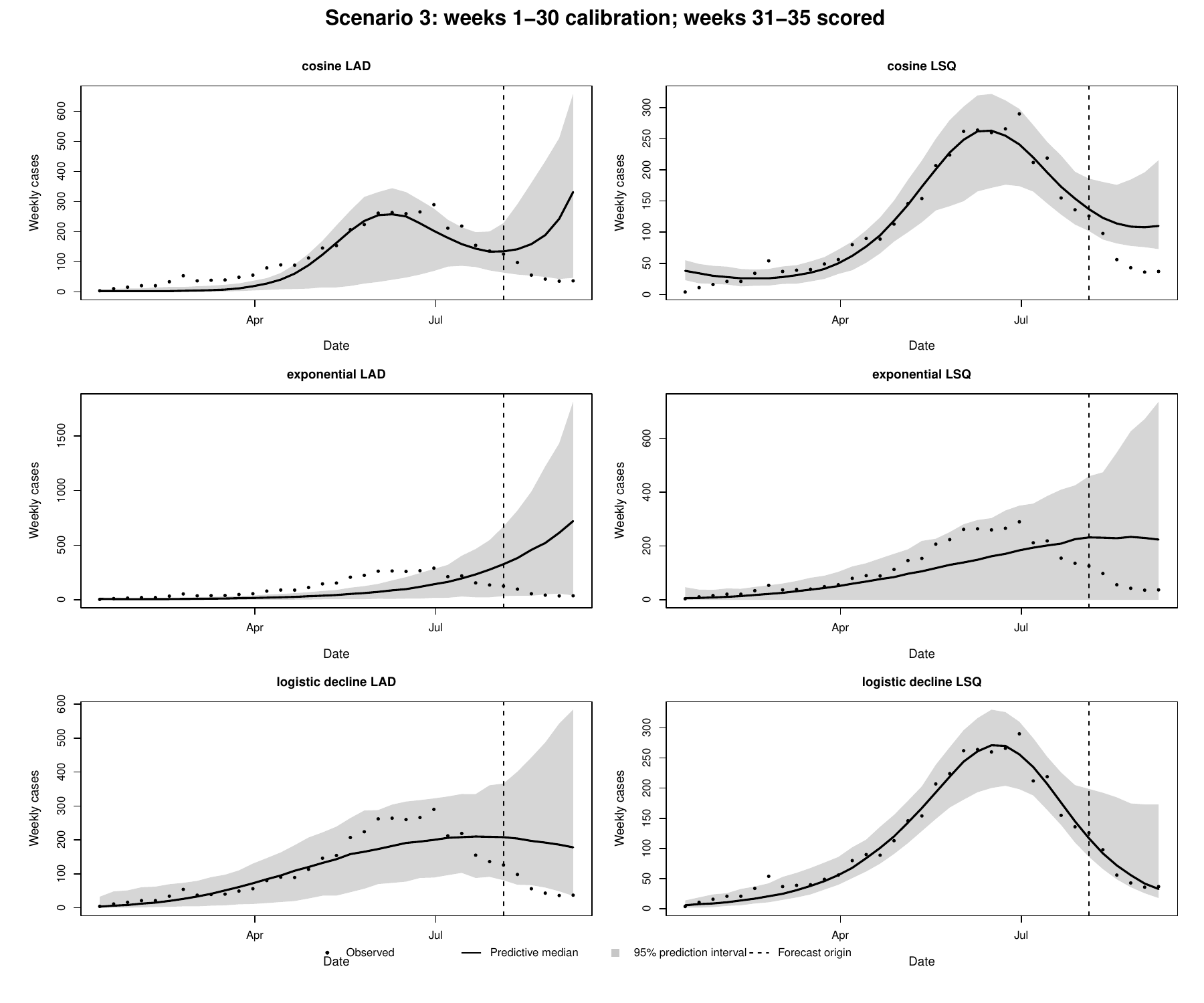}
\caption{Scenario 3: weeks 1--30 calibration and weeks 31--35 fixed-origin forecast comparison.}
\label{fig:fixed-s3}
\end{figure}

\begin{figure}[p]
\centering
\includegraphics[width=0.96\textwidth]{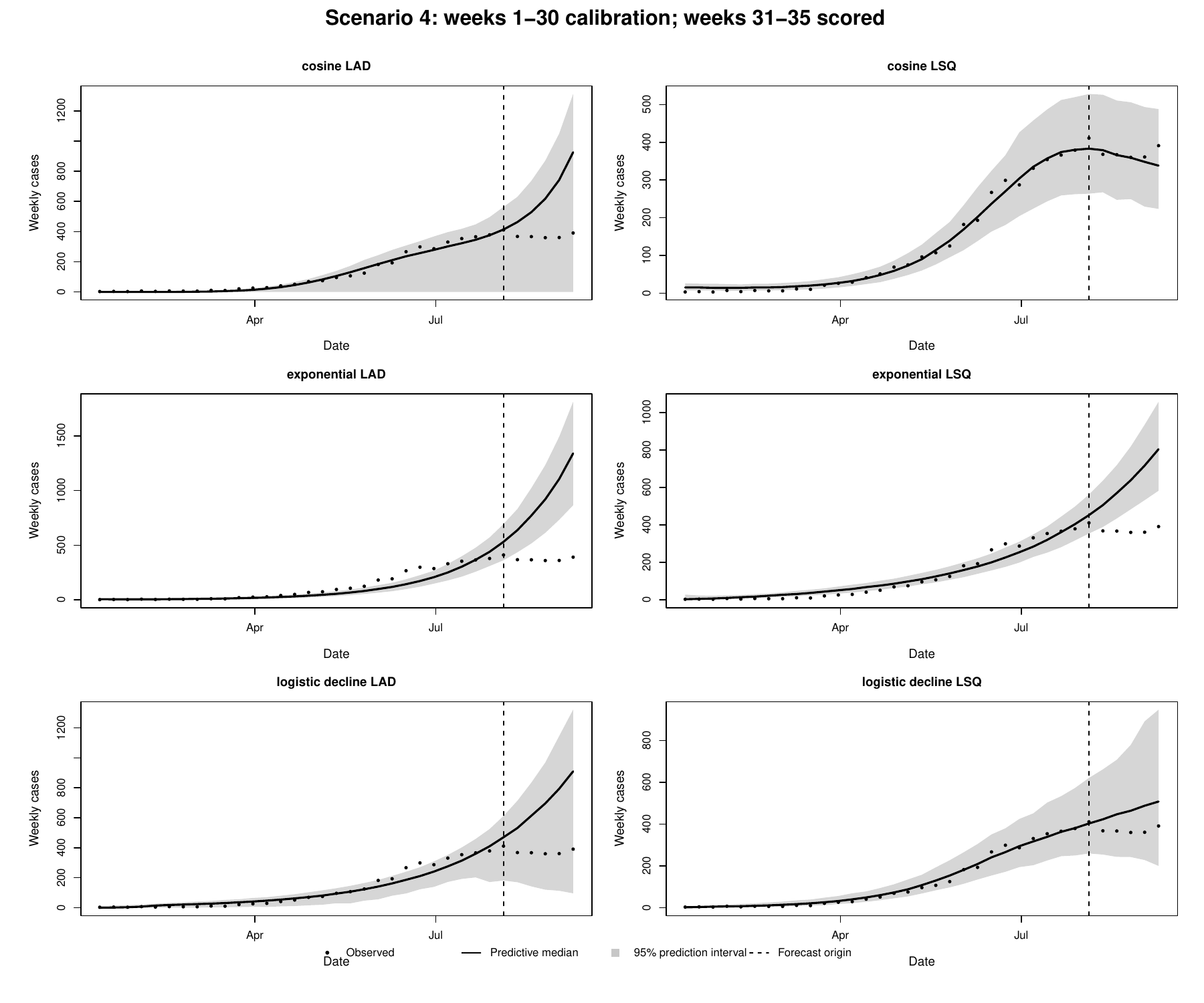}
\caption{Scenario 4: weeks 1--30 calibration and weeks 31--35 fixed-origin forecast comparison. The abrupt week-40 drop lies outside this scored window.}
\label{fig:fixed-s4}
\end{figure}

\section{Predictive uncertainty and identifiability diagnostics}
\label{sec:diagnostics}

\subsection{Uncertainty-method sensitivity}
The four fixed-origin bootstrap procedures produced materially different predictive distributions. Their averaged WIS, coverage, interval width, and minimum refit success are summarized in \Cref{tab:bootstrap-summary}.

\begin{table}[H]
\centering
\caption{Fixed-origin forecast uncertainty comparison, averaged over scenarios, transmission drivers, and losses. WIS denotes weighted interval score; Cov. denotes empirical coverage; Width denotes mean prediction-interval width; IID denotes independent and identically distributed residual resampling.}
\label{tab:bootstrap-summary}
\small
\begin{tabular}{@{}lrrrrrr@{}}
\toprule
Method & WIS & Cov.\ 80\% & Width 80\% & Cov.\ 95\% & Width 95\% & Min.\ refit success\\
\midrule
IID residual & 126.8 & 0.167 & 118.5 & 0.592 & 376.0 & 0.986\\
Wild & 121.4 & 0.175 & 138.7 & 0.708 & 394.9 & 0.991\\
Moving block & 130.5 & 0.192 & 124.6 & 0.642 & 365.5 & 0.992\\
Negative binomial & 107.6 & 0.383 & 209.1 & 0.767 & 455.4 & 0.988\\
\bottomrule
\end{tabular}
\end{table}

These averages are descriptive and should be interpreted with the scenario-specific coverage and randomized probability integral transform (PIT) diagnostics in the supplement, because uncertainty-method rankings vary by scenario and horizon. Averaged over scenarios, drivers, and losses, the negative-binomial parametric bootstrap achieved the lowest WIS and the highest 80\% and 95\% empirical coverage, with wider intervals reflecting its greater predictive dispersion. The wild bootstrap also improved 95\% coverage relative to IID residual resampling with a smaller increase in width. Horizon-specific diagnostics quantify where additional calibration is needed. These results show that the uncertainty mechanism is a substantive component of forecast evaluation and can change WIS rankings even when point forecasts are unchanged. The same uncertainty-method comparison is visualized in \Cref{fig:bootstrap-summary}.

\begin{figure}[t]
\centering
\includegraphics[width=0.90\textwidth]{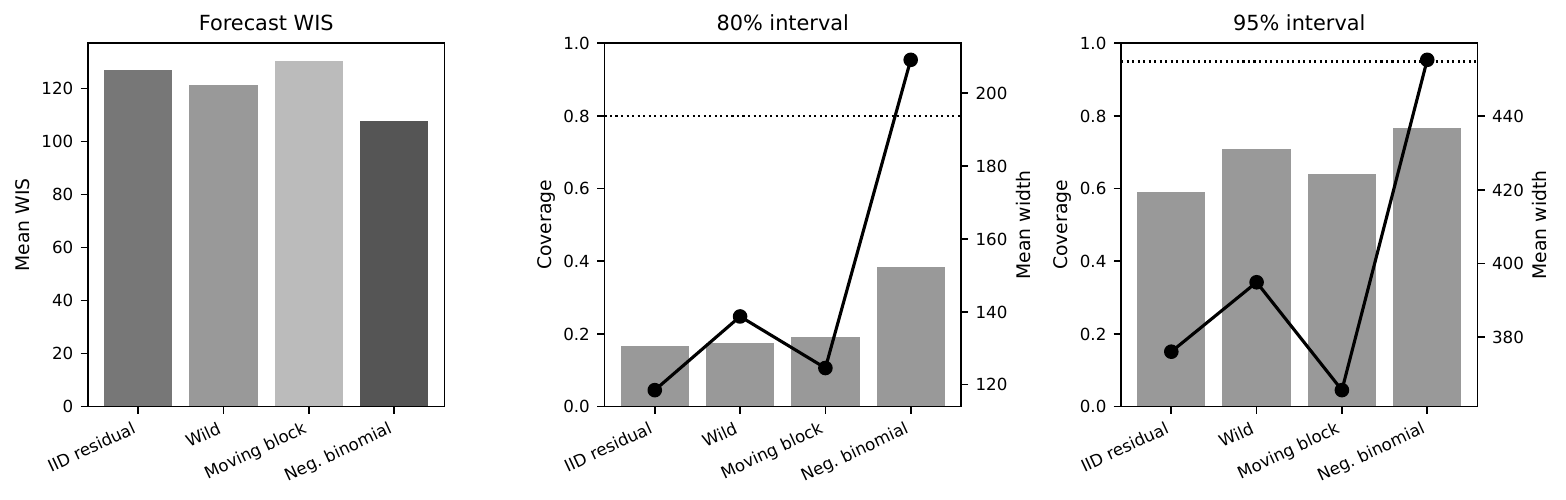}
\caption{Fixed-origin uncertainty-method comparison. WIS denotes weighted interval score; coverage is the empirical proportion of observations inside the stated predictive interval. Coverage and interval width are averaged over scenarios, transmission drivers, and loss functions.}
\label{fig:bootstrap-summary}
\end{figure}

Coverage-by-horizon and randomized PIT diagnostics are reported for each scenario in the supplement. The PIT distributions are used descriptively because the forecasts overlap across rolling origins and are not independent. Departures from uniformity, especially concentration near 0 or 1, identify systematic under- or overprediction that aggregate WIS alone can conceal.

\subsection{Sampled-output identifiability}
Sensitivity diagnostics were computed for all 24 scenario--driver--loss SEIR fits. The corresponding numerical-rank and condition-number summaries are reported in \Cref{tab:identifiability-summary}.

\begin{table}[t]
\centering
\caption{Sampled-output identifiability diagnostics for the 24 primary RAPIDD SEIR fits.}
\label{tab:identifiability-summary}
\small
\begin{tabular}{@{}lccc@{}}
\toprule
Transmission driver & Full numerical rank & Median $\log_{10}$ condition number & Range\\
\midrule
Cosine turning-point & 8/8 & 5.51 & 4.04--7.55\\
Exponential & 6/8 & 6.95 & 5.96--8.62\\
Logistic decline & 3/8 & 8.15 & 7.25--8.49\\
\bottomrule
\end{tabular}
\end{table}

The scaled sensitivity matrix was numerically full rank in 17 of 24 fits. All eight cosine fits were full rank, compared with six of eight exponential fits and three of eight logistic-decline fits. Among full-rank fits, condition numbers ranged from approximately $1.1\times10^4$ to $4.1\times10^8$, providing a quantitative scale for practical identifiability. The median condition number was about $3.2\times10^5$ for cosine, $8.9\times10^6$ for exponential, and $1.4\times10^8$ for logistic decline.

Profile objectives and a sensitivity analysis fixing $\sigma$ and $\gamma$ to Ebola-informed values are reported in the supplement. These diagnostics identify which parameters have localized objective support and which benefit from additional structural information. We therefore pair trajectory and forecast interpretation with sensitivity rank, conditioning, and profile localization, yielding an explicit criterion for stronger mechanistic interpretation of $\beta_0$, $\sigma$, $\gamma$, $E_0$, and $I_0$.

\section{Discussion}
\label{sec:discussion}
Calibration loss is part of the observation and validation strategy, not merely a numerical choice. Accordingly, we evaluate LAD as a robust, prespecified sensitivity-analysis option for time-varying SEIR models and characterize the settings in which it adds the greatest value alongside LSQ. The distinction is substantive. LAD and LSQ target different features of the conditional distribution, and their relative performance changes with the score, epidemic phase, reporting mechanism, transmission structure, and forecast horizon.

The phase-aware simulation identifies both the settings and the magnitude of the robustness gain. Across isolated spikes, backlog release, and temporary underreporting with delayed release, LAD was favored in 106 of 144 MAE comparisons and 104 of 144 WIS comparisons. Pooled across drivers, horizons, and replications, the median relative change in MAE during early growth was $-72.9\%$ under backlog release and $-84.2\%$ under temporary underreporting with delayed release. These mechanisms create influential residuals whose squared contribution can dominate LSQ. Clean negative-binomial and serially correlated settings serve as useful reference regimes, showing that the robustness gain is concentrated rather than automatic. The simulation therefore yields a clear practical rule: LAD adds the greatest value when isolated reporting anomalies would otherwise dominate squared-error calibration, especially during early growth and near the peak.

\subsection{Interpreting phase- and mechanism-dependent performance}
The Monte Carlo and RAPIDD findings are complementary because their designs emphasize different reporting regimes. The simulation deliberately isolates short, controlled reporting anomalies--isolated spikes, backlog release, and temporary underreporting with delayed release--so a few observations can exert disproportionate leverage under squared loss. LAD downweights those residuals and can improve subsequent forecasts. Many RAPIDD rolling windows contain sustained local rises or declines rather than isolated perturbations, so large recent residuals can carry information about continuing growth or changing control. This distinction explains why LAD often improves anomaly-focused simulations while LSQ is preferred within every transmission family in the RAPIDD rolling-origin analysis, and it demonstrates that the framework can identify the loss function most aligned with the active forecasting regime.

The RAPIDD results strengthen the methodological evaluation by imposing demanding external benchmarks. Across rolling origins, simple baselines had the smallest WIS in every scenario--horizon cell, and logistic-decline LSQ was consistently the strongest SEIR specification. The framework therefore separates epidemiological interpretability from short-horizon predictive ranking and embeds baseline comparison directly into mechanistic forecast assessment.

Transmission structure was at least as important as the loss function. The cosine curve is useful as a finite-window phenomenological turning-point function. The exponential curve is parsimonious but restrictive when data contain a peak and decline. The logistic-decline form directly represents gradual strengthening of control and provided the strongest SEIR rolling forecasts. Its additional parameters also produced larger condition numbers, illustrating the value of reporting predictive performance together with identifiability diagnostics. Predictive adequacy and mechanistic identifiability are distinct goals.

The uncertainty analysis also changes the interpretation of WIS. IID residual resampling assumes exchangeability and cannot represent mean-dependent variance or serial dependence. Wild and block procedures address different departures, but neither was uniformly superior. Negative-binomial sampling produced the best average WIS and coverage, with wider intervals reflecting a sharper trade-off between dispersion and calibration. The results reinforce the value of reporting point-error rankings together with calibrated predictive uncertainty.

The large-sample results specialize standard LAD arguments to differentiable ODE-implied interval outputs and establish existence, consistency, and asymptotic normality under an idealized increasing-reporting design with interior parameters, independent errors, stable local identification, and a positive error density at zero. The finite RAPIDD experiments, with 15–35 calibration observations, complement the theory by showing how the estimator behaves in short surveillance series and by using sensitivity matrices and profile objectives to quantify practical identification directly.

The broader methodological implication is that robustness cannot be judged independently of the estimand, reporting mechanism, and validation score. Although the numerical experiments are epidemic-specific, the workflow is portable to other nonlinear dynamical models: align model output with the observation interval; compare robust and nonrobust losses under plausible anomalies; construct uncertainty on the observation scale; benchmark forecasts out of sample; and pair mechanistic interpretation with explicit identifiability diagnostics. Material differences between LAD and a primary LSQ or count-likelihood analysis become informative sensitivity evidence about dependence on the observation model.

\section{Conclusion}
\label{sec:conclusion}
Least absolute deviations provides a theoretically grounded and computationally stable robustness analysis for interval-incidence calibration of time-varying SEIR models. Across the full phase-aware Monte Carlo design, LAD was favored in 56.7\% of MAE comparisons and 57.9\% of WIS comparisons; under isolated spikes, backlog release, and temporary underreporting with delayed release, these proportions increased to 73.6\% and 72.2\%, respectively. Near the epidemic peak, LAD was favored in 66.7\% of MAE and 71.7\% of WIS comparisons. All 120,000 primary simulation fits achieved nominal optimizer convergence. Rolling-origin RAPIDD validation further demonstrates that the framework can distinguish anomaly-driven robustness gains from sustained-trend forecasting regimes: logistic-decline LSQ was the strongest SEIR specification, while simple baselines provided a stringent overall WIS benchmark.

The practical implication is to evaluate loss function, transmission structure, uncertainty mechanism, and baseline forecasts jointly at epidemiologically relevant origins and horizons. LAD is especially compelling when a small number of reporting anomalies would otherwise dominate squared-error calibration, and its comparison with LSQ becomes an informative diagnostic when the active reporting regime is uncertain. Prespecifying LAD as a robustness analysis, then validating both point and probabilistic forecasts out of sample and pairing them with identifiability diagnostics, provides a rigorous and transferable workflow for nonlinear dynamical models. The accompanying R package, public data and code, automated validation, and interactive application make the full workflow directly reproducible and reusable.

\section*{Data, code, software, and interactive application availability}
The released RAPIDD scenario data, installable R package, full analysis code, fixed settings, deterministic seed scheme, processed results, tables, figures, automated validation, and detailed reproducibility instructions are publicly available at \url{https://github.com/YisaAdeniyiAbolade/lad-seir-calibration}. 

\section*{CRediT authorship contribution statement}
The author contributed to the conception, methodological development, analysis, interpretation, manuscript preparation, and critical review of the work. 

\section*{Declaration of competing interest}
The author declare that they have no known competing financial interests or personal relationships that could have appeared to influence the work reported in this paper.

\section*{Funding}
This research did not receive any specific grant from funding agencies in the public, commercial, or not-for-profit sectors.

\clearpage
\setcounter{figure}{0}
\setcounter{table}{0}
\setcounter{equation}{0}
\renewcommand{\theHfigure}{supp.\arabic{figure}}
\renewcommand{\theHtable}{supp.\arabic{table}}
\renewcommand{\theHequation}{supp.\arabic{equation}}
\newcommand{\coveragefig}[1]{\ifcase#1\or Figure_S1.pdf\or Figure_S3.pdf\or Figure_S5.pdf\or Figure_S7.pdf\fi}
\newcommand{\pitfig}[1]{\ifcase#1\or Figure_S2.pdf\or Figure_S4.pdf\or Figure_S6.pdf\or Figure_S8.pdf\fi}
\newcommand{\profilefig}[2]{\ifcase#1\or Figure_S9#2.pdf\or Figure_S10#2.pdf\or Figure_S11#2.pdf\or Figure_S12#2.pdf\fi}

\begin{center}
{\LARGE\bfseries Supplementary Material for Calibration of Time-Varying SEIR Models\par}
\vspace{1.2em}
{\large Yisa Abolade\textsuperscript{} \quad\textsuperscript{} \quad\textsuperscript{}\par}
\vspace{0.6em}
{\small
\textsuperscript{}Department of Mathematics and Statistics, Georgia State University, Atlanta, Georgia, USA\\
\texttt{yabolade1@gsu.edu} \quad \texttt{} \quad \texttt{}\\[0.4em]
\textsuperscript{} \texttt{}}
\end{center}
\vspace{1em}

\begin{abstract}
Complete proofs, numerical results, uncertainty and identifiability diagnostics, and fixed computational specifications supporting the main manuscript are provided here.
\end{abstract}
\noindent This supplementary material provides the complete proofs, full numerical results underlying the main-text summaries, uncertainty and identifiability diagnostics, and the fixed reproducibility specification.

\section*{Supplement roadmap}
\begin{itemize}
\item Appendix A contains the complete large-sample proofs.
\item Appendix B.1 describes the RAPIDD series and exact fixed- and rolling-origin windows.
\item Appendix B.2 gives the complete phase-aware simulation design, full LAD--LSQ paired results, and relative-difference summaries.
\item Appendix B.3 reports all rolling-origin and fixed-origin RAPIDD scores, including WIS skill ratios relative to the naive baseline.
\item Appendix B.4 contains scenario- and horizon-specific coverage and randomized PIT diagnostics.
\item Appendix B.5 reports uncertainty-method comparisons and bootstrap parameter intervals.
\item Appendix B.6 contains sensitivity singular values, condition numbers, profile objectives, and fixed-$\sigma$/$\gamma$ analyses.
\item Appendix C gives parameter bounds, starting-value rules, seed construction, optimizer settings, and software versions.
\end{itemize}
\appendix
\section{Proofs of the large-sample results}
\label{app:proofs}

\subsection{Finite-horizon global well-posedness}
\begin{proof}[Proof of the finite-horizon global well-posedness theorem]
Let $F(t,x;\theta)$ denote the right-hand side of the four SEIR state equations. Since $\beta(\cdot;\theta)$ is continuous on a finite interval, it is bounded. On every bounded subset of $\R^4$, the map $x\mapsto F(t,x;\theta)$ is Lipschitz uniformly in $t$; the only nonlinear term satisfies
\[
\abs{SI-\widetilde S\widetilde I}
\le \abs{S-\widetilde S}\abs{I}
+\abs{\widetilde S}\abs{I-\widetilde I}.
\]
The Picard--Lindelof theorem therefore gives a unique maximal local solution.

Summing the compartment equations gives
\[
\frac{d}{dt}(S+E+I+R)=0,
\]
so total population remains $N$. The vector field is quasi-positive on the boundary of the nonnegative orthant:
\[
S=0\Rightarrow\dot S=0,\quad
E=0\Rightarrow\dot E=\beta SI/N\ge0,
\]
\[
I=0\Rightarrow\dot I=\sigma E\ge0,\quad
R=0\Rightarrow\dot R=\gamma I\ge0.
\]
A trajectory starting in the epidemiological simplex cannot cross into a negative coordinate, and every state remains bounded by $N$. The solution therefore stays in a compact set on which the vector field is bounded and locally Lipschitz. The continuation theorem rules out finite-time blow-up and extends the solution through the complete finite horizon. Finally,
\[
0\le C(t)-C(0)=\int_0^t\sigma E(u)\,du\le \sigma Nt,
\]
so $C$ is nondecreasing and finite.
\end{proof}

\subsection{Parameter differentiability}
\begin{proof}[Proof of the parameter differentiability theorem]
Write the augmented state as $z=(S,E,I,R,C)\trans$ and its vector field as $G(t,z,\theta)$. The preceding result confines trajectories to a compact set. Continuous differentiability of $G$ and of the initial-condition map permits differentiation of
\[
z(t;\theta)=z_0(\theta)+\int_0^tG\{u,z(u;\theta),\theta\}\,du.
\]
The sensitivity matrix $Z_\theta(t)=\partial z(t;\theta)/\partial\theta\trans$ satisfies
\begin{equation}
\dot Z_\theta(t)=G_z\{t,z(t;\theta),\theta\}Z_\theta(t)
+G_\theta\{t,z(t;\theta),\theta\},
\qquad
Z_\theta(0)=\frac{\partial z_0(\theta)}{\partial\theta\trans}.
\label{eq:sensitivity-supp}
\end{equation}
This linear nonautonomous system has a unique solution because its coefficients are continuous and bounded. If $C_\theta(t)$ is the row corresponding to the incidence accumulator, then
\[
\frac{\partial\mu_j(\theta)}{\partial\theta\trans}
=C_\theta(\tau_j)-C_\theta(\tau_{j-1}).
\]
Continuity and uniformity follow from compactness of the parameter and time domains.
\end{proof}

\subsection{Existence, uniform convergence, and consistency}
\begin{proof}[Existence of constrained minimizers]
Every interval output is continuous in $\theta$, hence both $\abs{Y_j-\mu_j(\theta)}$ and $\{Y_j-\mu_j(\theta)\}^2$ are continuous. Their finite sums are continuous on compact $\mathcal H$, so LAD and LSQ minimizers exist by the Weierstrass theorem.
\end{proof}

\begin{proof}[Uniform convergence of the normalized LAD criterion]
Let
\[
d_{Tj}(\theta)=\mu_j(\theta_0)-\mu_j(\theta),\qquad
Z_{Tj}(\theta)=\abs{\varepsilon_{Tj}+d_{Tj}(\theta)}.
\]
The reverse triangle inequality and uniform Lipschitz assumption give
\begin{equation}
\abs{Z_{Tj}(\theta)-Z_{Tj}(\vartheta)}
\le K\norm{\theta-\vartheta}.
\label{eq:equi-supp}
\end{equation}
The same inequality holds after expectation. Fix $\delta>0$ and let $\{\theta_1,\ldots,\theta_M\}$ be a finite $\delta$-net of $\mathcal H$. For each $\theta$, choose a net point within $\delta$. Then
\begin{align*}
\sup_{\theta\in\mathcal H}\abs{\bar L_T(\theta)-Q_T(\theta)}
&\le 2K\delta\\
&\quad+\max_{1\le k\le M}
\abs{T^{-1}\sum_{j=1}^T
[Z_{Tj}(\theta_k)-\E Z_{Tj}(\theta_k)]}.
\end{align*}
Bounded model output and the uniform second-moment assumption imply a uniformly bounded variance for each net-point summand. Independence gives variance $O(T^{-1})$ for each average. Chebyshev's inequality and a union bound over the fixed finite net make the maximum $o_p(1)$. Let $T\to\infty$ and then $\delta\downarrow0$.
\end{proof}

\begin{proof}[Consistency of the LAD minimizer]
For fixed $\epsilon>0$, define
\[
\Delta_{T,\epsilon}=\inf_{\theta\in\mathcal H:\norm{\theta-\theta_0}\ge\epsilon}
\{Q_T(\theta)-Q_T(\theta_0)\}.
\]
Uniform separation gives $\Delta_{T,\epsilon}\ge\Delta_\epsilon>0$ for all sufficiently large $T$. On the event
\[
\sup_{\theta\in\mathcal H}\abs{\bar L_T(\theta)-Q_T(\theta)}<\Delta_\epsilon/3,
\]
every $\theta$ outside the $\epsilon$-ball has empirical criterion strictly larger than the criterion at $\theta_0$. Uniform convergence makes this event have probability tending to one. Hence every measurable minimizer converges in probability to $\theta_0$.
\end{proof}

\subsection{Asymptotic normality}
\begin{proof}[Proof of asymptotic normality]
For fixed $h\in\R^p$, let $\theta=\theta_0+h/\sqrt T$. Uniform differentiability gives
\[
\mu_j(\theta)-\mu_j(\theta_0)
=J_{Tj}\trans h/\sqrt T+q_{Tj}(h),
\]
where the aggregate remainder is $o_p(1)$ uniformly on compact sets. Knight's identity states that
\begin{equation}
\abs{u-v}-\abs{u}
=-v\operatorname{sign}(u)
+2\int_0^v\{\ind(u\le s)-\ind(u\le0)\}\,ds.
\label{eq:knight-supp}
\end{equation}
Applying \eqref{eq:knight-supp} with $u=\varepsilon_{Tj}$ and $v=J_{Tj}\trans h/\sqrt T$, the linear terms equal $-h\trans Z_T$, where
\[
Z_T=T^{-1/2}\sum_{j=1}^TJ_{Tj}\operatorname{sign}(\varepsilon_{Tj}).
\]
Median zero and absence of a point mass at zero imply mean-zero signs with variance one. The deterministic-array Lindeberg condition and $A_T\to A$ yield $Z_T\Rightarrow N(0,A)$.

Continuity of the error density at zero gives, uniformly for small $v$,
\[
\E\left[2\int_0^v\{\ind(\varepsilon\le s)-\ind(\varepsilon\le0)\}\,ds\right]
=f(0)v^2+o(v^2).
\]
The maximal-Jacobian and local-remainder assumptions make the centered aggregate fluctuation $o_p(1)$. The local objective therefore converges uniformly on compact sets to
\[
\mathcal V(h)=-h\trans Z+f(0)h\trans Ah,
\qquad Z\sim N(0,A).
\]
Since $A$ is positive definite, the unique minimizer is $h^*=(2f(0)A)^{-1}Z$. The convex argmin theorem and consistency give
\[
\sqrt T(\widehat\theta_T-\theta_0)\Rightarrow h^*,
\]
whose covariance is $(4f(0)^2)^{-1}A^{-1}$.
\end{proof}

\section{Complete numerical results}
\label{app:numerical}

\subsection{RAPIDD series and analysis windows}
Table~\ref{tab:s-scenario-summary} records the complete released-series summaries and explicitly separates the fixed calibration, scoring, and unused weeks.
\begin{table}[htbp]
\centering
\caption{RAPIDD scenario series and fixed-origin windows.}
\label{tab:s-scenario-summary}
\scriptsize
\begin{tabular}{@{}lrrrrlll@{}}
\toprule
Scenario & Weeks & Total & Max & Max week & Fixed calibration & Fixed score & Unused\\
\midrule
Scenario 1 & 46 & 5,765 & 359 & 24 (2014-06-23) & 1-30 & 31-35 & 36-46\\
Scenario 2 & 46 & 6,897 & 310 & 35 (2014-09-08) & 1-30 & 31-35 & 36-46\\
Scenario 3 & 46 & 3,980 & 290 & 25 (2014-06-30) & 1-30 & 31-35 & 36-46\\
Scenario 4 & 46 & 7,492 & 494 & 39 (2014-10-06) & 1-30 & 31-35 & 36-46\\
\bottomrule
\end{tabular}
\end{table}

\subsection{Phase-aware simulation design and full paired results}
The full 60-condition design is listed first, followed by parameter-recovery summaries and all paired phase--mechanism--driver--horizon differences.
\begin{longtable}{@{}rlll@{}}
\caption{Complete phase-aware Monte Carlo condition grid.}\label{tab:s-sim-design}\\
\toprule
ID & Driver & Epidemic phase & Reporting mechanism\\
\midrule
\endfirsthead
\multicolumn{4}{c}{\tablename\ \thetable\ (continued)}\\
\toprule
ID & Driver & Epidemic phase & Reporting mechanism\\
\midrule
\endhead
\bottomrule
\endfoot
1 & cosine & early\_growth & clean\\
2 & exponential & early\_growth & clean\\
3 & logistic\_decline & early\_growth & clean\\
4 & cosine & near\_peak & clean\\
5 & exponential & near\_peak & clean\\
6 & logistic\_decline & near\_peak & clean\\
7 & cosine & early\_decline & clean\\
8 & exponential & early\_decline & clean\\
\end{longtable}
\begin{table}[htbp]
\centering
\caption{Parameter-recovery summary across simulation conditions.}
\label{tab:s-param-summary}
\small
\begin{tabular}{@{}lrrrr@{}}
\toprule
Parameter & Conditions & LAD lower RMSE & Mean RMSE LAD & Mean RMSE LSQ\\
\midrule
E0 & 60 & 56 & 3,185.9 & 6,729.6\\
I0 & 60 & 38 & 7,411.4 & 7,830.3\\
a & 40 & 9 & 0.886 & 0.855\\
b & 20 & 9 & 0.028 & 0.028\\
beta0 & 60 & 26 & 0.681 & 0.674\\
gamma & 60 & 15 & 0.180 & 0.168\\
k & 20 & 15 & 0.103 & 0.169\\
omega & 20 & 20 & 0.106 & 0.154\\
q & 20 & 18 & 0.165 & 0.249\\
sigma & 60 & 48 & 0.087 & 0.099\\
tau & 20 & 4 & 39.60 & 40.19\\
\bottomrule
\end{tabular}
\end{table}
\FloatBarrier
\begin{table}[p]
\centering
\small
\caption{Magnitude of paired LAD-versus-LSQ forecast differences. Entries are the median percentage change in MAE, $100(\mathrm{MAE}_{\mathrm{LAD}}-\mathrm{MAE}_{\mathrm{LSQ}})/\mathrm{MAE}_{\mathrm{LSQ}}$, with the interquartile range (IQR) in parentheses, pooled over drivers, horizons, and replications. Negative values favor LAD. Ratios are omitted when the matched LSQ MAE is zero (416 of 240,000 comparisons).}
\label{tab:s-relative-mae}
\begin{tabular}{llr}
\toprule
Reporting mechanism & Epidemic phase & Median \% change (IQR) \\
\midrule
Clean & Early growth & 6.1 (-40.5, 106.2) \\
 & Near peak & 4.2 (-42.2, 91.4) \\
 & Early decline & 0.0 (-42.7, 78.0) \\
 & Late decline & 0.0 (-53.2, 114.2) \\
\addlinespace
Isolated spikes & Early growth & 7.8 (-44.0, 126.0) \\
 & Near peak & -7.1 (-55.5, 79.0) \\
 & Early decline & -13.8 (-59.1, 70.5) \\
 & Late decline & -11.0 (-62.7, 71.4) \\
\addlinespace
Backlog release & Early growth & -72.9 (-89.1, -21.6) \\
 & Near peak & -45.0 (-79.0, 41.8) \\
 & Early decline & -15.4 (-60.6, 72.0) \\
 & Late decline & 9.2 (-50.0, 146.7) \\
\addlinespace
Temporary underreporting with delayed release & Early growth & -84.2 (-92.9, -65.4) \\
 & Near peak & -56.2 (-78.5, -10.2) \\
 & Early decline & -27.5 (-74.8, 67.1) \\
 & Late decline & 14.5 (-48.9, 160.7) \\
\addlinespace
Serially correlated reporting error & Early growth & 10.5 (-39.1, 114.1) \\
 & Near peak & 0.0 (-48.7, 79.0) \\
 & Early decline & -6.8 (-50.8, 62.4) \\
 & Late decline & 0.0 (-52.9, 107.7) \\
\addlinespace
\bottomrule
\end{tabular}
\end{table}

\begin{landscape}
\scriptsize
\begin{longtable}{@{}llllrrrcr@{}}
\caption{Complete paired Monte Carlo forecast differences. Negative differences favor LAD.}\label{tab:s-paired-forecast}\\
\toprule
Driver & Phase & Mechanism & Metric & Horizon & LAD-LSQ & 95\% MC CI & Prop.\ LAD lower & $R$\\
\midrule
\endfirsthead
\multicolumn{9}{c}{\tablename\ \thetable\ (continued)}\\
\toprule
Driver & Phase & Mechanism & Metric & Horizon & LAD-LSQ & 95\% MC CI & Prop.\ LAD lower & $R$\\
\midrule
\endhead
\bottomrule
\endfoot
cosine & early\_decline & backlog\_release & MAE & 1 & -437.0 & [-610.8, -263.2] & 0.556 & 1000\\
cosine & early\_decline & backlog\_release & WIS & 1 & -243.2 & [-369.4, -116.9] & 0.545 & 1000\\
cosine & early\_decline & backlog\_release & MAE & 3 & -1,021.8 & [-1,200.9, -842.8] & 0.630 & 1000\\
cosine & early\_decline & backlog\_release & WIS & 3 & -659.6 & [-792.8, -526.4] & 0.621 & 1000\\
cosine & early\_decline & backlog\_release & MAE & 5 & -1,413.5 & [-1,589.5, -1,237.5] & 0.693 & 1000\\
cosine & early\_decline & backlog\_release & WIS & 5 & -960.8 & [-1,093.8, -827.7] & 0.683 & 1000\\
cosine & early\_decline & backlog\_release & MAE & 10 & -1,213.3 & [-1,363.9, -1,062.7] & 0.722 & 1000\\
cosine & early\_decline & backlog\_release & WIS & 10 & -853.1 & [-966.2, -740.0] & 0.715 & 1000\\
\end{longtable}
\end{landscape}

\clearpage
\subsection{Complete rolling-origin and fixed-origin results}
Tables~\ref{tab:s-rolling-full} and \ref{tab:s-fixed-full} report all point and interval scores used to construct the concise main-text winner summaries.
\begin{table}[H]
\centering
\caption{Rolling-origin WIS skill ratios for the strongest SEIR specification, logistic-decline LSQ, relative to the naive last-observation baseline. A ratio below 1 favors the SEIR model; a ratio above 1 quantifies its forecasting penalty relative to the naive baseline.}
\label{tab:s-wis-skill-ratios}
\begin{tabular}{lccccc}
\toprule
Scenario & $h=1$ & $h=2$ & $h=3$ & $h=4$ & $h=5$ \\
\midrule
Scenario 1 & 1.47 & 1.34 & 1.20 & 1.33 & 1.46 \\
Scenario 2 & 0.90 & 1.03 & 1.06 & 1.17 & 1.46 \\
Scenario 3 & 1.25 & 1.12 & 1.14 & 1.18 & 1.23 \\
Scenario 4 & 1.13 & 1.07 & 1.11 & 1.24 & 1.25 \\
\bottomrule
\end{tabular}
\end{table}
\begin{landscape}
\scriptsize
\begin{longtable}{@{}lllrrrrrrrrrr@{}}
\caption{Complete rolling-origin summaries across origins 15--35. MAE denotes mean absolute error; WIS denotes weighted interval score; Cov denotes empirical coverage; W denotes mean interval width; and $h$ denotes forecast horizon.}\label{tab:s-rolling-full}\\
\toprule
Scenario & Model & $h$ & MAE & WIS & Cov50 & W50 & Cov80 & W80 & Cov90 & W90 & Cov95 & W95\\
\midrule
\endfirsthead
\multicolumn{13}{c}{\tablename\ \thetable\ (continued)}\\
\toprule
Scenario & Model & $h$ & MAE & WIS & Cov50 & W50 & Cov80 & W80 & Cov90 & W90 & Cov95 & W95\\
\midrule
\endhead
\bottomrule
\endfoot
scenario1 & cosine\_LAD & 1 & 81.33 & 60.51 & 0.381 & 60.41 & 0.524 & 124.3 & 0.524 & 167.4 & 0.667 & 204.7\\
scenario1 & cosine\_LSQ & 1 & 44.79 & 28.81 & 0.571 & 56.30 & 0.714 & 111.5 & 0.810 & 146.7 & 0.905 & 180.5\\
scenario1 & exponential\_LAD & 1 & 131.6 & 101.5 & 0.286 & 77.51 & 0.381 & 157.4 & 0.571 & 210.7 & 0.571 & 261.8\\
scenario1 & exponential\_LSQ & 1 & 102.3 & 77.71 & 0.286 & 60.83 & 0.333 & 127.0 & 0.476 & 172.2 & 0.476 & 212.3\\
scenario1 & logistic\_decline\_LAD & 1 & 45.64 & 30.81 & 0.476 & 62.02 & 0.667 & 123.5 & 0.810 & 162.8 & 0.952 & 201.4\\
scenario1 & logistic\_decline\_LSQ & 1 & 37.64 & 24.44 & 0.429 & 55.34 & 0.714 & 109.0 & 0.762 & 145.3 & 0.905 & 176.4\\
scenario1 & naive\_last & 1 & 27.00 & 16.62 & 0.381 & 43.27 & 0.905 & 88.66 & 0.952 & 119.7 & 1.000 & 148.2\\
scenario1 & recent\_exponential & 1 & 29.05 & 20.33 & 0.667 & 51.47 & 0.857 & 101.2 & 0.905 & 132.1 & 0.905 & 161.0\\
scenario1 & cosine\_LAD & 2 & 127.6 & 97.12 & 0.238 & 69.75 & 0.429 & 143.4 & 0.429 & 191.3 & 0.476 & 235.0\\
scenario1 & cosine\_LSQ & 2 & 78.10 & 53.32 & 0.238 & 63.24 & 0.571 & 123.7 & 0.667 & 161.0 & 0.714 & 197.9\\
scenario1 & exponential\_LAD & 2 & 180.6 & 141.2 & 0.238 & 87.60 & 0.381 & 181.3 & 0.429 & 242.2 & 0.524 & 298.5\\
scenario1 & exponential\_LSQ & 2 & 134.1 & 104.7 & 0.238 & 67.74 & 0.333 & 141.9 & 0.429 & 191.3 & 0.429 & 237.4\\
scenario1 & logistic\_decline\_LAD & 2 & 63.71 & 44.34 & 0.429 & 66.01 & 0.571 & 131.6 & 0.619 & 173.9 & 0.714 & 212.8\\
scenario1 & logistic\_decline\_LSQ & 2 & 56.24 & 38.34 & 0.381 & 58.65 & 0.524 & 116.4 & 0.571 & 152.9 & 0.762 & 184.9\\
scenario1 & naive\_last & 2 & 44.83 & 28.66 & 0.286 & 42.82 & 0.571 & 89.04 & 0.667 & 119.8 & 0.714 & 148.4\\
scenario1 & recent\_exponential & 2 & 44.24 & 31.87 & 0.571 & 56.76 & 0.810 & 111.5 & 0.857 & 144.8 & 0.857 & 174.4\\
scenario1 & cosine\_LAD & 3 & 177.8 & 142.2 & 0.238 & 78.21 & 0.333 & 162.8 & 0.381 & 218.9 & 0.381 & 269.0\\
scenario1 & cosine\_LSQ & 3 & 107.6 & 80.28 & 0.286 & 68.72 & 0.429 & 134.7 & 0.476 & 175.5 & 0.571 & 214.7\\
scenario1 & exponential\_LAD & 3 & 243.0 & 194.8 & 0.095 & 99.60 & 0.238 & 204.1 & 0.381 & 273.3 & 0.429 & 334.8\\
scenario1 & exponential\_LSQ & 3 & 173.6 & 137.4 & 0.190 & 77.79 & 0.286 & 158.5 & 0.333 & 213.9 & 0.429 & 266.1\\
scenario1 & logistic\_decline\_LAD & 3 & 89.62 & 63.86 & 0.333 & 70.86 & 0.429 & 141.2 & 0.524 & 185.9 & 0.619 & 228.8\\
scenario1 & logistic\_decline\_LSQ & 3 & 73.81 & 52.96 & 0.429 & 63.23 & 0.524 & 122.0 & 0.524 & 162.0 & 0.524 & 198.2\\
scenario1 & naive\_last & 3 & 63.86 & 44.14 & 0.143 & 43.23 & 0.286 & 89.38 & 0.429 & 120.7 & 0.571 & 146.1\\
scenario1 & recent\_exponential & 3 & 67.45 & 49.30 & 0.381 & 61.74 & 0.714 & 120.9 & 0.810 & 158.2 & 0.810 & 189.3\\
scenario1 & cosine\_LAD & 4 & 261.2 & 212.1 & 0.143 & 95.69 & 0.190 & 197.8 & 0.238 & 265.5 & 0.238 & 324.4\\
scenario1 & cosine\_LSQ & 4 & 156.8 & 119.0 & 0.048 & 80.76 & 0.190 & 155.5 & 0.429 & 205.4 & 0.476 & 249.3\\
scenario1 & exponential\_LAD & 4 & 316.0 & 256.3 & 0.143 & 116.7 & 0.143 & 239.0 & 0.190 & 318.2 & 0.333 & 390.4\\
scenario1 & exponential\_LSQ & 4 & 226.2 & 183.3 & 0.095 & 89.81 & 0.238 & 186.9 & 0.238 & 248.9 & 0.286 & 305.8\\
scenario1 & logistic\_decline\_LAD & 4 & 118.2 & 88.40 & 0.238 & 77.32 & 0.333 & 150.2 & 0.476 & 195.1 & 0.476 & 237.5\\
scenario1 & logistic\_decline\_LSQ & 4 & 104.6 & 78.01 & 0.238 & 67.07 & 0.333 & 132.0 & 0.429 & 172.2 & 0.476 & 208.9\\
scenario1 & naive\_last & 4 & 81.24 & 58.56 & 0.048 & 43.68 & 0.238 & 90.10 & 0.286 & 121.8 & 0.524 & 151.3\\
scenario1 & recent\_exponential & 4 & 99.57 & 75.08 & 0.381 & 68.13 & 0.524 & 134.2 & 0.667 & 174.0 & 0.714 & 209.3\\
scenario1 & cosine\_LAD & 5 & 355.6 & 296.6 & 0.095 & 111.1 & 0.143 & 231.2 & 0.143 & 311.7 & 0.190 & 381.9\\
scenario1 & cosine\_LSQ & 5 & 219.0 & 173.3 & 0.048 & 92.45 & 0.190 & 180.4 & 0.238 & 236.4 & 0.381 & 287.8\\
scenario1 & exponential\_LAD & 5 & 405.4 & 334.1 & 0.048 & 133.0 & 0.095 & 274.8 & 0.143 & 363.1 & 0.286 & 445.6\\
scenario1 & exponential\_LSQ & 5 & 302.2 & 248.2 & 0.095 & 100.8 & 0.143 & 217.9 & 0.190 & 293.7 & 0.238 & 362.3\\
scenario1 & logistic\_decline\_LAD & 5 & 149.4 & 112.7 & 0.095 & 82.47 & 0.286 & 161.0 & 0.333 & 210.4 & 0.476 & 258.0\\
scenario1 & logistic\_decline\_LSQ & 5 & 133.0 & 102.3 & 0.143 & 72.43 & 0.333 & 139.4 & 0.381 & 182.7 & 0.429 & 219.5\\
scenario1 & naive\_last & 5 & 92.95 & 70.15 & 0.048 & 43.09 & 0.143 & 89.71 & 0.286 & 120.7 & 0.381 & 148.4\\
scenario1 & recent\_exponential & 5 & 135.3 & 106.3 & 0.333 & 75.94 & 0.524 & 148.9 & 0.524 & 192.3 & 0.619 & 231.0\\
scenario2 & cosine\_LAD & 1 & 43.19 & 29.36 & 0.429 & 46.06 & 0.619 & 91.61 & 0.714 & 120.8 & 0.810 & 148.4\\
scenario2 & cosine\_LSQ & 1 & 32.21 & 20.29 & 0.524 & 52.34 & 0.762 & 104.0 & 0.857 & 136.0 & 1.000 & 164.9\\
scenario2 & exponential\_LAD & 1 & 54.50 & 35.81 & 0.190 & 43.09 & 0.333 & 88.54 & 0.571 & 118.6 & 0.667 & 148.1\\
scenario2 & exponential\_LSQ & 1 & 32.02 & 21.01 & 0.476 & 42.87 & 0.667 & 87.04 & 0.762 & 118.7 & 0.905 & 147.9\\
scenario2 & logistic\_decline\_LAD & 1 & 31.07 & 19.31 & 0.429 & 41.27 & 0.619 & 83.43 & 0.952 & 110.8 & 1.000 & 136.3\\
scenario2 & logistic\_decline\_LSQ & 1 & 18.79 & 13.85 & 0.619 & 43.27 & 0.762 & 87.70 & 0.905 & 114.9 & 0.952 & 139.2\\
scenario2 & naive\_last & 1 & 24.81 & 15.46 & 0.333 & 33.39 & 0.810 & 70.53 & 0.952 & 96.76 & 0.952 & 121.5\\
scenario2 & recent\_exponential & 1 & 19.12 & 12.51 & 0.667 & 42.25 & 0.905 & 85.97 & 0.905 & 115.4 & 0.952 & 142.4\\
scenario2 & cosine\_LAD & 2 & 67.02 & 46.60 & 0.286 & 53.67 & 0.524 & 108.2 & 0.667 & 142.3 & 0.762 & 175.8\\
scenario2 & cosine\_LSQ & 2 & 48.55 & 29.69 & 0.429 & 59.60 & 0.667 & 117.6 & 0.762 & 152.3 & 0.905 & 183.3\\
scenario2 & exponential\_LAD & 2 & 73.05 & 50.81 & 0.143 & 48.68 & 0.381 & 99.98 & 0.476 & 135.9 & 0.619 & 170.4\\
scenario2 & exponential\_LSQ & 2 & 43.38 & 27.40 & 0.381 & 47.54 & 0.667 & 97.14 & 0.667 & 131.3 & 0.952 & 163.2\\
scenario2 & logistic\_decline\_LAD & 2 & 41.67 & 27.24 & 0.476 & 43.22 & 0.619 & 86.50 & 0.714 & 114.6 & 0.857 & 140.9\\
scenario2 & logistic\_decline\_LSQ & 2 & 25.05 & 17.36 & 0.619 & 45.33 & 0.810 & 91.57 & 0.905 & 121.3 & 0.952 & 148.6\\
scenario2 & naive\_last & 2 & 27.45 & 16.91 & 0.381 & 33.83 & 0.619 & 70.40 & 0.810 & 97.86 & 0.857 & 121.4\\
scenario2 & recent\_exponential & 2 & 25.62 & 16.57 & 0.524 & 44.62 & 0.810 & 92.01 & 0.952 & 123.0 & 0.952 & 152.9\\
scenario2 & cosine\_LAD & 3 & 92.55 & 67.70 & 0.238 & 60.76 & 0.381 & 122.8 & 0.476 & 162.5 & 0.524 & 199.1\\
scenario2 & cosine\_LSQ & 3 & 75.24 & 51.45 & 0.381 & 67.39 & 0.571 & 133.7 & 0.619 & 174.0 & 0.619 & 209.0\\
scenario2 & exponential\_LAD & 3 & 105.2 & 76.75 & 0.143 & 55.44 & 0.238 & 113.3 & 0.286 & 154.2 & 0.381 & 190.8\\
scenario2 & exponential\_LSQ & 3 & 63.93 & 42.48 & 0.238 & 53.13 & 0.476 & 108.6 & 0.619 & 146.2 & 0.619 & 181.8\\
scenario2 & logistic\_decline\_LAD & 3 & 60.14 & 42.88 & 0.238 & 45.63 & 0.476 & 91.65 & 0.571 & 121.6 & 0.667 & 146.7\\
scenario2 & logistic\_decline\_LSQ & 3 & 38.07 & 25.65 & 0.429 & 47.62 & 0.619 & 96.28 & 0.714 & 128.1 & 0.762 & 154.4\\
scenario2 & naive\_last & 3 & 34.38 & 24.20 & 0.381 & 33.44 & 0.571 & 70.97 & 0.619 & 97.50 & 0.762 & 122.3\\
scenario2 & recent\_exponential & 3 & 40.71 & 25.56 & 0.286 & 48.73 & 0.571 & 100.9 & 0.762 & 134.2 & 0.905 & 165.4\\
scenario2 & cosine\_LAD & 4 & 135.5 & 102.5 & 0.190 & 70.09 & 0.286 & 140.1 & 0.286 & 186.2 & 0.333 & 226.0\\
scenario2 & cosine\_LSQ & 4 & 125.7 & 91.11 & 0.190 & 79.05 & 0.429 & 155.8 & 0.524 & 203.9 & 0.524 & 244.6\\
scenario2 & exponential\_LAD & 4 & 140.0 & 106.2 & 0.095 & 63.12 & 0.190 & 129.5 & 0.190 & 177.1 & 0.238 & 220.1\\
scenario2 & exponential\_LSQ & 4 & 86.19 & 59.38 & 0.048 & 58.21 & 0.476 & 118.8 & 0.524 & 161.9 & 0.571 & 200.9\\
scenario2 & logistic\_decline\_LAD & 4 & 78.55 & 57.80 & 0.143 & 48.44 & 0.286 & 95.86 & 0.429 & 127.1 & 0.524 & 152.7\\
scenario2 & logistic\_decline\_LSQ & 4 & 56.21 & 37.76 & 0.238 & 50.51 & 0.476 & 100.9 & 0.571 & 133.6 & 0.714 & 161.3\\
scenario2 & naive\_last & 4 & 46.10 & 32.36 & 0.143 & 33.83 & 0.476 & 70.99 & 0.571 & 97.83 & 0.667 & 123.6\\
scenario2 & recent\_exponential & 4 & 60.31 & 39.83 & 0.286 & 53.55 & 0.429 & 109.0 & 0.667 & 149.1 & 0.762 & 181.3\\
scenario2 & cosine\_LAD & 5 & 201.5 & 158.2 & 0.095 & 82.74 & 0.190 & 168.9 & 0.190 & 223.9 & 0.286 & 273.5\\
scenario2 & cosine\_LSQ & 5 & 173.1 & 132.6 & 0.143 & 92.75 & 0.333 & 181.6 & 0.476 & 237.0 & 0.476 & 286.1\\
scenario2 & exponential\_LAD & 5 & 186.7 & 147.3 & 0.048 & 72.94 & 0.190 & 148.1 & 0.190 & 200.7 & 0.238 & 250.9\\
scenario2 & exponential\_LSQ & 5 & 114.7 & 82.72 & 0.095 & 65.23 & 0.190 & 133.3 & 0.429 & 179.7 & 0.524 & 222.4\\
scenario2 & logistic\_decline\_LAD & 5 & 98.36 & 73.44 & 0.095 & 51.13 & 0.190 & 101.4 & 0.238 & 133.6 & 0.571 & 161.6\\
scenario2 & logistic\_decline\_LSQ & 5 & 75.26 & 53.82 & 0.143 & 53.95 & 0.429 & 107.1 & 0.524 & 139.7 & 0.524 & 169.4\\
scenario2 & naive\_last & 5 & 51.12 & 36.76 & 0.190 & 33.46 & 0.381 & 71.09 & 0.476 & 97.77 & 0.714 & 122.5\\
scenario2 & recent\_exponential & 5 & 82.02 & 57.94 & 0.286 & 57.75 & 0.381 & 119.5 & 0.429 & 161.4 & 0.524 & 198.3\\
scenario3 & cosine\_LAD & 1 & 79.60 & 60.04 & 0.143 & 44.79 & 0.476 & 92.35 & 0.476 & 123.1 & 0.476 & 152.6\\
scenario3 & cosine\_LSQ & 1 & 63.36 & 47.48 & 0.333 & 47.44 & 0.619 & 95.63 & 0.714 & 126.9 & 0.810 & 156.5\\
scenario3 & exponential\_LAD & 1 & 119.4 & 90.48 & 0.143 & 65.06 & 0.333 & 132.5 & 0.429 & 177.9 & 0.524 & 218.1\\
scenario3 & exponential\_LSQ & 1 & 103.6 & 76.67 & 0.143 & 57.77 & 0.333 & 120.3 & 0.381 & 165.2 & 0.381 & 205.4\\
scenario3 & logistic\_decline\_LAD & 1 & 45.12 & 31.06 & 0.429 & 47.66 & 0.619 & 96.25 & 0.619 & 128.5 & 0.762 & 158.0\\
scenario3 & logistic\_decline\_LSQ & 1 & 27.95 & 18.66 & 0.476 & 44.92 & 0.857 & 89.40 & 0.905 & 117.2 & 0.952 & 143.2\\
scenario3 & naive\_last & 1 & 22.71 & 14.98 & 0.619 & 36.65 & 0.714 & 76.83 & 0.857 & 104.9 & 1.000 & 130.2\\
scenario3 & recent\_exponential & 1 & 23.93 & 15.43 & 0.476 & 44.51 & 0.857 & 88.45 & 0.905 & 115.2 & 1.000 & 140.4\\
scenario3 & cosine\_LAD & 2 & 113.4 & 91.62 & 0.238 & 51.21 & 0.381 & 103.9 & 0.381 & 139.7 & 0.381 & 171.8\\
scenario3 & cosine\_LSQ & 2 & 99.26 & 75.46 & 0.095 & 51.66 & 0.476 & 105.1 & 0.476 & 139.7 & 0.571 & 170.1\\
scenario3 & exponential\_LAD & 2 & 159.6 & 125.7 & 0.190 & 71.39 & 0.238 & 147.6 & 0.286 & 196.3 & 0.381 & 243.6\\
scenario3 & exponential\_LSQ & 2 & 145.0 & 111.7 & 0.095 & 63.15 & 0.190 & 128.9 & 0.286 & 176.5 & 0.286 & 224.3\\
scenario3 & logistic\_decline\_LAD & 2 & 62.02 & 42.66 & 0.143 & 49.17 & 0.429 & 99.63 & 0.524 & 131.7 & 0.619 & 161.1\\
scenario3 & logistic\_decline\_LSQ & 2 & 46.31 & 30.89 & 0.333 & 45.54 & 0.619 & 90.89 & 0.762 & 119.3 & 0.810 & 144.3\\
scenario3 & naive\_last & 2 & 41.26 & 27.51 & 0.190 & 36.10 & 0.429 & 76.63 & 0.619 & 104.8 & 0.667 & 132.8\\
scenario3 & recent\_exponential & 2 & 40.74 & 25.91 & 0.381 & 46.81 & 0.524 & 91.70 & 0.667 & 120.2 & 0.857 & 144.4\\
scenario3 & cosine\_LAD & 3 & 158.0 & 131.6 & 0.286 & 58.50 & 0.286 & 119.0 & 0.333 & 158.7 & 0.333 & 195.3\\
scenario3 & cosine\_LSQ & 3 & 125.3 & 100.3 & 0.143 & 57.49 & 0.381 & 115.6 & 0.524 & 154.6 & 0.571 & 190.4\\
scenario3 & exponential\_LAD & 3 & 207.2 & 167.7 & 0.143 & 78.98 & 0.238 & 162.0 & 0.238 & 221.3 & 0.238 & 272.1\\
scenario3 & exponential\_LSQ & 3 & 191.4 & 153.9 & 0.095 & 67.51 & 0.143 & 141.9 & 0.190 & 194.8 & 0.238 & 243.2\\
scenario3 & logistic\_decline\_LAD & 3 & 76.81 & 54.58 & 0.095 & 51.16 & 0.238 & 102.8 & 0.524 & 134.9 & 0.524 & 163.4\\
scenario3 & logistic\_decline\_LSQ & 3 & 68.79 & 48.04 & 0.190 & 48.44 & 0.333 & 94.86 & 0.524 & 124.6 & 0.571 & 149.7\\
scenario3 & naive\_last & 3 & 58.52 & 42.13 & 0.095 & 36.10 & 0.333 & 76.11 & 0.333 & 105.1 & 0.381 & 131.6\\
scenario3 & recent\_exponential & 3 & 63.60 & 43.56 & 0.238 & 50.38 & 0.333 & 97.29 & 0.429 & 127.2 & 0.571 & 155.4\\
scenario3 & cosine\_LAD & 4 & 216.4 & 180.6 & 0.048 & 67.60 & 0.143 & 137.9 & 0.286 & 186.9 & 0.333 & 229.9\\
scenario3 & cosine\_LSQ & 4 & 177.6 & 145.8 & 0.048 & 66.52 & 0.286 & 132.5 & 0.333 & 176.8 & 0.476 & 215.0\\
scenario3 & exponential\_LAD & 4 & 261.6 & 216.0 & 0.095 & 88.75 & 0.190 & 182.5 & 0.190 & 246.0 & 0.190 & 306.4\\
scenario3 & exponential\_LSQ & 4 & 243.2 & 201.9 & 0.095 & 75.85 & 0.143 & 161.0 & 0.143 & 222.2 & 0.190 & 272.0\\
scenario3 & logistic\_decline\_LAD & 4 & 97.93 & 72.57 & 0.048 & 52.35 & 0.190 & 106.1 & 0.286 & 140.2 & 0.286 & 170.4\\
scenario3 & logistic\_decline\_LSQ & 4 & 90.69 & 67.25 & 0.095 & 50.71 & 0.286 & 100.2 & 0.429 & 131.6 & 0.524 & 159.4\\
scenario3 & naive\_last & 4 & 76.12 & 57.05 & 0.000 & 36.42 & 0.238 & 77.31 & 0.286 & 104.7 & 0.286 & 130.3\\
scenario3 & recent\_exponential & 4 & 86.67 & 65.09 & 0.190 & 54.73 & 0.238 & 106.2 & 0.286 & 140.2 & 0.381 & 167.6\\
scenario3 & cosine\_LAD & 5 & 273.7 & 234.3 & 0.143 & 77.38 & 0.190 & 160.7 & 0.238 & 217.6 & 0.333 & 266.8\\
scenario3 & cosine\_LSQ & 5 & 224.2 & 188.6 & 0.095 & 75.33 & 0.238 & 151.4 & 0.286 & 200.0 & 0.476 & 245.9\\
scenario3 & exponential\_LAD & 5 & 327.3 & 272.2 & 0.000 & 101.6 & 0.048 & 207.4 & 0.143 & 277.6 & 0.190 & 345.5\\
scenario3 & exponential\_LSQ & 5 & 306.6 & 258.5 & 0.048 & 82.00 & 0.095 & 178.4 & 0.143 & 245.3 & 0.190 & 304.3\\
scenario3 & logistic\_decline\_LAD & 5 & 117.7 & 91.29 & 0.048 & 53.96 & 0.143 & 108.8 & 0.238 & 145.9 & 0.238 & 177.5\\
scenario3 & logistic\_decline\_LSQ & 5 & 112.3 & 87.01 & 0.143 & 53.80 & 0.238 & 106.4 & 0.333 & 138.5 & 0.429 & 166.0\\
scenario3 & naive\_last & 5 & 91.60 & 70.83 & 0.000 & 35.62 & 0.095 & 76.23 & 0.095 & 104.3 & 0.190 & 131.5\\
scenario3 & recent\_exponential & 5 & 118.9 & 92.95 & 0.095 & 60.04 & 0.238 & 117.1 & 0.286 & 153.1 & 0.333 & 184.1\\
scenario4 & cosine\_LAD & 1 & 93.00 & 68.52 & 0.286 & 57.79 & 0.476 & 125.9 & 0.667 & 174.7 & 0.762 & 217.3\\
scenario4 & cosine\_LSQ & 1 & 65.45 & 42.42 & 0.333 & 70.86 & 0.762 & 140.4 & 0.857 & 184.0 & 0.857 & 222.4\\
scenario4 & exponential\_LAD & 1 & 121.9 & 91.80 & 0.286 & 61.96 & 0.476 & 132.7 & 0.524 & 183.7 & 0.667 & 233.8\\
scenario4 & exponential\_LSQ & 1 & 63.83 & 44.62 & 0.476 & 57.07 & 0.714 & 120.7 & 0.857 & 164.6 & 0.857 & 203.9\\
scenario4 & logistic\_decline\_LAD & 1 & 45.69 & 28.69 & 0.476 & 62.89 & 0.857 & 132.3 & 0.952 & 178.5 & 0.952 & 217.0\\
scenario4 & logistic\_decline\_LSQ & 1 & 28.43 & 18.29 & 0.619 & 62.90 & 0.905 & 125.7 & 1.000 & 167.0 & 1.000 & 203.4\\
scenario4 & naive\_last & 1 & 24.90 & 16.24 & 0.476 & 48.40 & 0.857 & 104.4 & 0.905 & 141.9 & 0.952 & 176.5\\
scenario4 & recent\_exponential & 1 & 27.60 & 18.43 & 0.619 & 60.86 & 0.905 & 124.9 & 1.000 & 165.8 & 1.000 & 202.7\\
scenario4 & cosine\_LAD & 2 & 167.6 & 132.4 & 0.286 & 70.85 & 0.476 & 154.6 & 0.476 & 213.0 & 0.619 & 267.4\\
scenario4 & cosine\_LSQ & 2 & 110.2 & 76.25 & 0.333 & 84.58 & 0.429 & 167.3 & 0.667 & 220.3 & 0.762 & 264.6\\
scenario4 & exponential\_LAD & 2 & 175.1 & 136.8 & 0.286 & 73.74 & 0.381 & 159.9 & 0.429 & 221.0 & 0.524 & 279.2\\
scenario4 & exponential\_LSQ & 2 & 92.17 & 68.17 & 0.286 & 66.01 & 0.571 & 139.4 & 0.762 & 190.5 & 0.810 & 236.7\\
scenario4 & logistic\_decline\_LAD & 2 & 73.81 & 46.67 & 0.286 & 72.70 & 0.571 & 150.1 & 0.810 & 199.0 & 0.952 & 239.2\\
scenario4 & logistic\_decline\_LSQ & 2 & 52.05 & 31.27 & 0.381 & 67.53 & 0.714 & 137.9 & 0.905 & 182.3 & 1.000 & 222.1\\
scenario4 & naive\_last & 2 & 44.95 & 29.30 & 0.238 & 48.67 & 0.524 & 104.2 & 0.714 & 142.6 & 0.762 & 176.4\\
scenario4 & recent\_exponential & 2 & 53.12 & 32.83 & 0.333 & 67.21 & 0.762 & 137.0 & 0.857 & 183.0 & 0.952 & 224.5\\
scenario4 & cosine\_LAD & 3 & 255.0 & 208.3 & 0.238 & 81.06 & 0.381 & 183.5 & 0.381 & 257.7 & 0.381 & 323.3\\
scenario4 & cosine\_LSQ & 3 & 173.8 & 128.9 & 0.286 & 102.5 & 0.429 & 201.1 & 0.429 & 263.6 & 0.571 & 319.4\\
scenario4 & exponential\_LAD & 3 & 247.6 & 198.1 & 0.190 & 86.45 & 0.333 & 188.4 & 0.333 & 265.4 & 0.381 & 337.4\\
scenario4 & exponential\_LSQ & 3 & 131.9 & 98.68 & 0.286 & 77.15 & 0.524 & 162.5 & 0.571 & 222.8 & 0.762 & 274.2\\
scenario4 & logistic\_decline\_LAD & 3 & 106.8 & 72.17 & 0.238 & 78.00 & 0.524 & 167.5 & 0.571 & 221.9 & 0.714 & 268.9\\
scenario4 & logistic\_decline\_LSQ & 3 & 77.55 & 49.81 & 0.286 & 74.46 & 0.476 & 148.9 & 0.619 & 198.3 & 0.714 & 239.2\\
scenario4 & naive\_last & 3 & 64.14 & 45.04 & 0.190 & 49.22 & 0.429 & 104.9 & 0.571 & 140.7 & 0.619 & 173.5\\
scenario4 & recent\_exponential & 3 & 80.24 & 51.37 & 0.190 & 74.41 & 0.619 & 152.6 & 0.714 & 204.0 & 0.810 & 249.7\\
scenario4 & cosine\_LAD & 4 & 335.7 & 281.2 & 0.190 & 91.49 & 0.286 & 212.4 & 0.286 & 297.0 & 0.286 & 373.2\\
scenario4 & cosine\_LSQ & 4 & 281.0 & 219.8 & 0.095 & 125.3 & 0.333 & 249.6 & 0.429 & 326.4 & 0.476 & 394.1\\
scenario4 & exponential\_LAD & 4 & 339.0 & 278.2 & 0.143 & 102.4 & 0.190 & 222.2 & 0.286 & 310.3 & 0.286 & 395.7\\
scenario4 & exponential\_LSQ & 4 & 184.6 & 144.1 & 0.238 & 89.58 & 0.333 & 188.5 & 0.476 & 256.0 & 0.571 & 321.2\\
scenario4 & logistic\_decline\_LAD & 4 & 152.0 & 109.9 & 0.143 & 88.79 & 0.381 & 187.0 & 0.524 & 248.9 & 0.619 & 298.3\\
scenario4 & logistic\_decline\_LSQ & 4 & 111.9 & 76.29 & 0.190 & 82.70 & 0.429 & 167.2 & 0.476 & 220.5 & 0.619 & 265.6\\
scenario4 & naive\_last & 4 & 82.81 & 61.70 & 0.190 & 48.40 & 0.286 & 104.1 & 0.381 & 141.8 & 0.619 & 177.1\\
scenario4 & recent\_exponential & 4 & 117.6 & 78.93 & 0.095 & 84.79 & 0.429 & 174.5 & 0.524 & 232.2 & 0.667 & 283.6\\
scenario4 & cosine\_LAD & 5 & 451.5 & 385.0 & 0.095 & 110.8 & 0.286 & 252.5 & 0.286 & 353.0 & 0.286 & 436.0\\
scenario4 & cosine\_LSQ & 5 & 434.5 & 356.6 & 0.190 & 162.3 & 0.286 & 314.4 & 0.429 & 413.5 & 0.429 & 500.5\\
scenario4 & exponential\_LAD & 5 & 473.2 & 398.1 & 0.095 & 123.3 & 0.190 & 269.3 & 0.190 & 377.3 & 0.238 & 480.2\\
scenario4 & exponential\_LSQ & 5 & 279.4 & 222.2 & 0.095 & 103.3 & 0.190 & 217.0 & 0.238 & 298.6 & 0.381 & 371.0\\
scenario4 & logistic\_decline\_LAD & 5 & 212.1 & 163.6 & 0.095 & 98.20 & 0.333 & 206.3 & 0.333 & 276.2 & 0.524 & 334.4\\
scenario4 & logistic\_decline\_LSQ & 5 & 155.5 & 113.9 & 0.190 & 91.72 & 0.286 & 181.2 & 0.429 & 236.5 & 0.524 & 288.7\\
scenario4 & naive\_last & 5 & 112.9 & 91.23 & 0.190 & 49.23 & 0.286 & 104.5 & 0.333 & 141.9 & 0.429 & 174.8\\
scenario4 & recent\_exponential & 5 & 164.5 & 120.6 & 0.143 & 94.99 & 0.286 & 194.1 & 0.429 & 260.6 & 0.619 & 316.8\\
\end{longtable}
\end{landscape}
\begin{landscape}
\scriptsize
\begin{longtable}{@{}lllrrrrrrrrrrr@{}}
\caption{Prespecified fixed-origin results: weeks 1--30 calibration and weeks 31--35 scoring.}\label{tab:s-fixed-full}\\
\toprule
Scenario & Model & Type & MSE & MAE & WIS & C50 & W50 & C80 & W80 & C90 & W90 & C95 & W95\\
\midrule
\endfirsthead
\multicolumn{14}{c}{\tablename\ \thetable\ (continued)}\\
\toprule
Scenario & Model & Type & MSE & MAE & WIS & C50 & W50 & C80 & W80 & C90 & W90 & C95 & W95\\
\midrule
\endhead
\bottomrule
\endfoot
scenario1 & naive\_last & baseline & 1,449.4 & 33.80 & 20.91 & 0.200 & 38.52 & 0.600 & 81.11 & 0.800 & 109.6 & 1.000 & 135.5\\
scenario1 & recent\_exponential & baseline & 575.8 & 20.60 & 12.75 & 0.400 & 26.20 & 0.800 & 52.60 & 0.800 & 70.59 & 0.800 & 87.92\\
scenario1 & cosine\_LAD & seir & 35,951 & 183.4 & 136.6 & 0.000 & 76.98 & 0.000 & 163.8 & 0.000 & 229.0 & 0.200 & 292.4\\
scenario1 & cosine\_LSQ & seir & 489.6 & 17.60 & 11.23 & 0.600 & 31.20 & 0.800 & 62.53 & 1.000 & 84.19 & 1.000 & 108.9\\
scenario1 & exponential\_LAD & seir & 69,917 & 257.6 & 211.8 & 0.000 & 75.43 & 0.000 & 161.9 & 0.000 & 222.2 & 0.000 & 280.4\\
scenario1 & exponential\_LSQ & seir & 48,364 & 215.3 & 171.7 & 0.000 & 72.23 & 0.000 & 152.7 & 0.000 & 207.0 & 0.000 & 257.0\\
scenario1 & logistic\_decline\_LAD & seir & 11,690 & 107.2 & 73.69 & 0.000 & 57.23 & 0.000 & 121.5 & 0.000 & 163.0 & 0.200 & 201.6\\
scenario1 & logistic\_decline\_LSQ & seir & 9,609.2 & 97.10 & 66.22 & 0.000 & 53.92 & 0.000 & 112.4 & 0.000 & 156.3 & 0.200 & 201.2\\
scenario2 & naive\_last & baseline & 2,563.2 & 39.10 & 26.74 & 0.400 & 36.52 & 0.600 & 78.05 & 0.600 & 112.1 & 1.000 & 140.4\\
scenario2 & recent\_exponential & baseline & 1,925.8 & 35.40 & 21.85 & 0.600 & 51.60 & 0.600 & 104.6 & 1.000 & 138.4 & 1.000 & 170.6\\
scenario2 & cosine\_LAD & seir & 3,393.2 & 41.60 & 28.74 & 0.600 & 51.87 & 0.600 & 110.2 & 0.600 & 146.1 & 0.800 & 180.5\\
scenario2 & cosine\_LSQ & seir & 676.4 & 21.60 & 14.60 & 0.800 & 60.92 & 1.000 & 124.2 & 1.000 & 165.8 & 1.000 & 199.8\\
scenario2 & exponential\_LAD & seir & 14,495 & 113.0 & 83.24 & 0.000 & 48.87 & 0.000 & 100.5 & 0.000 & 141.2 & 0.400 & 182.7\\
scenario2 & exponential\_LSQ & seir & 535.6 & 21.10 & 12.73 & 0.400 & 44.75 & 1.000 & 97.58 & 1.000 & 132.9 & 1.000 & 168.2\\
scenario2 & logistic\_decline\_LAD & seir & 11,141 & 96.50 & 75.47 & 0.000 & 36.23 & 0.000 & 72.85 & 0.400 & 96.32 & 0.400 & 118.3\\
scenario2 & logistic\_decline\_LSQ & seir & 6,182.2 & 68.10 & 48.30 & 0.000 & 37.55 & 0.400 & 81.98 & 0.400 & 112.5 & 0.600 & 140.8\\
scenario3 & naive\_last & baseline & 5,571.8 & 71.00 & 53.17 & 0.000 & 32.60 & 0.200 & 67.83 & 0.200 & 92.06 & 0.200 & 114.9\\
scenario3 & recent\_exponential & baseline & 498.4 & 19.90 & 13.18 & 0.200 & 19.52 & 0.200 & 36.40 & 0.400 & 47.66 & 0.600 & 58.40\\
scenario3 & cosine\_LAD & seir & 10,129 & 83.30 & 66.14 & 0.200 & 38.00 & 0.200 & 73.25 & 0.200 & 97.33 & 0.200 & 122.5\\
scenario3 & cosine\_LSQ & seir & 6,483.8 & 75.80 & 62.82 & 0.000 & 23.72 & 0.000 & 47.93 & 0.000 & 64.06 & 0.200 & 78.73\\
scenario3 & exponential\_LAD & seir & 64,377 & 250.4 & 212.3 & 0.000 & 69.23 & 0.000 & 139.5 & 0.000 & 185.1 & 0.000 & 225.2\\
scenario3 & exponential\_LSQ & seir & 31,002 & 173.8 & 135.9 & 0.000 & 66.67 & 0.000 & 139.0 & 0.000 & 188.3 & 0.000 & 230.2\\
scenario3 & logistic\_decline\_LAD & seir & 1,818.5 & 39.50 & 25.51 & 0.000 & 18.60 & 0.000 & 38.53 & 0.600 & 53.78 & 0.800 & 75.25\\
scenario3 & logistic\_decline\_LSQ & seir & 603.2 & 22.00 & 13.60 & 0.400 & 24.20 & 0.400 & 47.73 & 0.600 & 63.86 & 1.000 & 77.73\\
scenario4 & naive\_last & baseline & 1,583.8 & 37.80 & 21.85 & 0.200 & 71.12 & 1.000 & 163.4 & 1.000 & 224.9 & 1.000 & 276.9\\
scenario4 & recent\_exponential & baseline & 16,955 & 125.0 & 73.85 & 0.000 & 129.5 & 0.400 & 238.7 & 0.600 & 301.6 & 1.000 & 359.0\\
scenario4 & cosine\_LAD & seir & 164,801 & 375.4 & 301.6 & 0.000 & 94.12 & 0.000 & 258.3 & 0.000 & 370.6 & 0.000 & 483.1\\
scenario4 & cosine\_LSQ & seir & 8,864.4 & 89.20 & 50.36 & 0.200 & 118.5 & 0.600 & 226.9 & 1.000 & 293.4 & 1.000 & 352.4\\
scenario4 & exponential\_LAD & seir & 414,692 & 595.3 & 525.3 & 0.000 & 83.23 & 0.000 & 194.4 & 0.000 & 302.1 & 0.000 & 423.8\\
scenario4 & exponential\_LSQ & seir & 104,905 & 306.0 & 240.0 & 0.000 & 79.43 & 0.000 & 202.7 & 0.000 & 323.0 & 0.000 & 411.2\\
scenario4 & logistic\_decline\_LAD & seir & 18,957 & 134.8 & 81.90 & 0.000 & 115.3 & 0.200 & 235.8 & 0.400 & 306.5 & 1.000 & 365.8\\
scenario4 & logistic\_decline\_LSQ & seir & 1,438.2 & 37.10 & 22.17 & 1.000 & 94.32 & 1.000 & 186.1 & 1.000 & 239.2 & 1.000 & 288.6\\
\end{longtable}
\end{landscape}

\clearpage
\subsection{Coverage and randomized PIT diagnostics}
The following scenario-specific pages show how interval calibration changes with forecast horizon and whether randomized PIT values reveal systematic under- or overprediction.
\foreach \s in {1,2,3,4}{%
\begin{figure}[p]
\centering
\includegraphics[width=0.94\textwidth]{Supplementary_Figures/\coveragefig{\s}}
\caption{Scenario \s: empirical interval coverage by rolling forecast horizon for the primary negative-binomial predictive distributions. Horizontal reference lines denote nominal levels.}
\label{fig:s-coverage-\s}
\end{figure}
\begin{figure}[p]
\centering
\includegraphics[width=0.94\textwidth]{Supplementary_Figures/\pitfig{\s}}
\caption{Scenario \s: randomized probability integral transform (PIT) diagnostics for rolling-origin forecasts. Uniformity is the calibration reference.}
\label{fig:s-pit-\s}
\end{figure}
}
\FloatBarrier

\clearpage
\subsection{Uncertainty-method sensitivity and bootstrap parameter intervals}
The next tables give all fixed-origin scores for the four uncertainty procedures and all percentile parameter intervals from successful refits.
\begin{landscape}
\scriptsize
\begin{longtable}{@{}llllrrrrrrrrrr@{}}
\caption{Fixed-origin uncertainty-method comparison over the five scored weeks. WIS denotes weighted interval score; Cov denotes empirical coverage; Width denotes mean interval width; IID denotes independent and identically distributed residual resampling; NB denotes negative binomial.}\label{tab:s-bootstrap-comparison}\\
\toprule
Scenario & Driver & Loss & Method & MAE & WIS & C50 & W50 & C80 & W80 & C90 & W90 & C95 & W95\\
\midrule
\endfirsthead
\multicolumn{14}{c}{\tablename\ \thetable\ (continued)}\\
\toprule
Scenario & Driver & Loss & Method & MAE & WIS & C50 & W50 & C80 & W80 & C90 & W90 & C95 & W95\\
\midrule
\endhead
\bottomrule
\endfoot
scenario1 & cosine & LAD & iid\_residual & 118.8 & 74.99 & 0.000 & 55.54 & 0.000 & 106.3 & 1.000 & 357.0 & 1.000 & 614.7\\
scenario1 & cosine & LAD & moving\_block & 118.4 & 81.09 & 0.000 & 49.82 & 0.000 & 87.72 & 0.000 & 116.6 & 1.000 & 472.3\\
scenario1 & cosine & LAD & negative\_binomial & 120.8 & 86.33 & 0.000 & 21.32 & 0.000 & 71.88 & 1.000 & 478.8 & 1.000 & 730.1\\
scenario1 & cosine & LAD & wild & 119.2 & 79.80 & 0.000 & 56.74 & 0.000 & 106.3 & 1.000 & 354.3 & 1.000 & 681.2\\
scenario1 & cosine & LSQ & iid\_residual & 19.26 & 14.32 & 0.400 & 29.09 & 0.600 & 67.21 & 1.000 & 190.5 & 1.000 & 387.4\\
scenario1 & cosine & LSQ & moving\_block & 20.04 & 13.99 & 0.400 & 31.17 & 0.800 & 74.63 & 1.000 & 111.4 & 1.000 & 355.0\\
scenario1 & cosine & LSQ & negative\_binomial & 18.40 & 13.53 & 0.600 & 35.12 & 0.800 & 77.98 & 1.000 & 210.5 & 1.000 & 397.7\\
scenario1 & cosine & LSQ & wild & 18.97 & 13.06 & 0.600 & 28.03 & 0.800 & 71.52 & 1.000 & 99.90 & 1.000 & 387.5\\
\end{longtable}
\end{landscape}
\begin{landscape}
\scriptsize
\begin{longtable}{@{}llllllrrrr@{}}
\caption{Complete fixed-origin bootstrap parameter intervals.}\label{tab:s-bootstrap-intervals}\\
\toprule
Scenario & Driver & Loss & Method & Parameter & Estimate & Lower & Median & Upper & Successful draws\\
\midrule
\endfirsthead
\multicolumn{10}{c}{\tablename\ \thetable\ (continued)}\\
\toprule
Scenario & Driver & Loss & Method & Parameter & Estimate & Lower & Median & Upper & Successful draws\\
\midrule
\endhead
\bottomrule
\endfoot
scenario1 & cosine & LAD & iid\_residual & E0 & 28.78 & 1.00e-06 & 28.78 & 61.98 & 1997\\
scenario1 & cosine & LAD & iid\_residual & I0 & 2.95e-04 & 1.00e-06 & 2.95e-04 & 95.35 & 1997\\
scenario1 & cosine & LAD & iid\_residual & a & -0.949 & -0.950 & -0.949 & 0.674 & 1997\\
scenario1 & cosine & LAD & iid\_residual & beta0 & 0.656 & 1.00e-04 & 0.656 & 0.661 & 1997\\
scenario1 & cosine & LAD & iid\_residual & gamma & 0.499 & 0.071 & 0.499 & 0.500 & 1997\\
scenario1 & cosine & LAD & iid\_residual & omega & 0.030 & 0.021 & 0.030 & 0.284 & 1997\\
scenario1 & cosine & LAD & iid\_residual & sigma & 0.049 & 0.048 & 0.049 & 0.293 & 1997\\
scenario1 & cosine & LAD & moving\_block & E0 & 28.78 & 8.62e-05 & 28.78 & 50.24 & 1000\\
\end{longtable}
\end{landscape}

\clearpage
\subsection{Identifiability diagnostics and fixed-rate sensitivity}
Complete sensitivity ranks and condition numbers are followed by one profile-objective page per scenario and the analysis fixing progression and removal rates.
\begin{table}[p]
\centering
\caption{Complete sampled-output identifiability diagnostics.}
\label{tab:s-ident-full}
\scriptsize
\begin{tabular}{@{}lllrrrrr@{}}
\toprule
Scenario & Driver & Loss & Free $p$ & Rank & Smallest singular value & Condition number & $\log_{10}$ condition\\
\midrule
scenario1 & cosine & LAD & 7 & 7 & 13.95 & 22,851 & 4.36\\
scenario1 & cosine & LSQ & 7 & 7 & 0.294 & 344,971 & 5.54\\
scenario1 & exponential & LAD & 7 & 7 & 0.070 & 906,811 & 5.96\\
scenario1 & exponential & LSQ & 7 & 7 & 0.039 & 20,289,177 & 7.31\\
scenario1 & logistic\_decline & LAD & 8 & 8 & 0.096 & 65,053,837 & 7.81\\
scenario1 & logistic\_decline & LSQ & 8 & 6 & 0.012 & 309,679,239 & 8.49\\
scenario2 & cosine & LAD & 7 & 7 & 0.364 & 1,082,466 & 6.03\\
scenario2 & cosine & LSQ & 7 & 7 & 0.020 & 35,169,765 & 7.55\\
scenario2 & exponential & LAD & 7 & 7 & 0.024 & 1,254,272 & 6.10\\
scenario2 & exponential & LSQ & 7 & 6 & 0.010 & 100,799,438 & 8.00\\
scenario2 & logistic\_decline & LAD & 8 & 7 & 0.015 & 172,203,347 & 8.24\\
scenario2 & logistic\_decline & LSQ & 8 & 7 & 0.008 & 116,802,123 & 8.07\\
scenario3 & cosine & LAD & 7 & 7 & 0.931 & 1,743,012 & 6.24\\
scenario3 & cosine & LSQ & 7 & 7 & 4.96 & 32,137 & 4.51\\
scenario3 & exponential & LAD & 7 & 7 & 0.017 & 3,239,593 & 6.51\\
scenario3 & exponential & LSQ & 7 & 7 & 0.029 & 3,904,442 & 6.59\\
scenario3 & logistic\_decline & LAD & 8 & 8 & 0.010 & 17,831,405 & 7.25\\
scenario3 & logistic\_decline & LSQ & 8 & 8 & 0.018 & 71,491,240 & 7.85\\
scenario4 & cosine & LAD & 7 & 7 & 28.23 & 10,931 & 4.04\\
scenario4 & cosine & LSQ & 7 & 7 & 0.256 & 298,455 & 5.47\\
scenario4 & exponential & LAD & 7 & 6 & 0.045 & 412,938,011 & 8.62\\
scenario4 & exponential & LSQ & 7 & 7 & 0.021 & 32,484,563 & 7.51\\
scenario4 & logistic\_decline & LAD & 8 & 6 & 0.040 & 252,415,884 & 8.40\\
scenario4 & logistic\_decline & LSQ & 8 & 7 & 0.003 & 203,415,644 & 8.31\\
\bottomrule
\end{tabular}
\end{table}
\FloatBarrier

\foreach \s in {1,2,3,4}{%
\clearpage
\begin{figure}[p]
\centering
\begin{subfigure}{0.32\textwidth}
\includegraphics[width=\linewidth]{Supplementary_Figures/\profilefig{\s}{a}}
\caption{Cosine LAD}
\end{subfigure}\hfill
\begin{subfigure}{0.32\textwidth}
\includegraphics[width=\linewidth]{Supplementary_Figures/\profilefig{\s}{b}}
\caption{Cosine LSQ}
\end{subfigure}\hfill
\begin{subfigure}{0.32\textwidth}
\includegraphics[width=\linewidth]{Supplementary_Figures/\profilefig{\s}{c}}
\caption{Exponential LAD}
\end{subfigure}

\medskip
\begin{subfigure}{0.32\textwidth}
\includegraphics[width=\linewidth]{Supplementary_Figures/\profilefig{\s}{d}}
\caption{Exponential LSQ}
\end{subfigure}\hfill
\begin{subfigure}{0.32\textwidth}
\includegraphics[width=\linewidth]{Supplementary_Figures/\profilefig{\s}{e}}
\caption{Logistic decline LAD}
\end{subfigure}\hfill
\begin{subfigure}{0.32\textwidth}
\includegraphics[width=\linewidth]{Supplementary_Figures/\profilefig{\s}{f}}
\caption{Logistic decline LSQ}
\end{subfigure}
\caption{Scenario \s: selected normalized profile objectives. Flat, asymmetric, or boundary-directed profiles indicate weak practical identification.}
\label{fig:s-profile-\s}
\end{figure}
\clearpage
}

\begin{landscape}
\scriptsize
\begin{longtable}{@{}llllrrrrrrrr@{}}
\caption{Sensitivity analysis with estimated versus fixed Ebola-informed progression and removal rates.}\label{tab:s-fixed-rate}\\
\toprule
Scenario & Driver & Loss & Rates & MAE & WIS & $\beta_0$ & $\sigma$ & $\gamma$ & $E_0$ & $I_0$ & Conv.\\
\midrule
\endfirsthead
\multicolumn{12}{c}{\tablename\ \thetable\ (continued)}\\
\toprule
Scenario & Driver & Loss & Rates & MAE & WIS & $\beta_0$ & $\sigma$ & $\gamma$ & $E_0$ & $I_0$ & Conv.\\
\midrule
\endhead
\bottomrule
\endfoot
scenario1 & cosine & LAD & estimated & 453.0 & 387.9 & 0.175 & 0.143 & 0.132 & 18.51 & 7.12e-04 & 0\\
scenario1 & cosine & LAD & fixed\_Ebola\_informed & 151.2 & 123.0 & 0.154 & 0.088 & 0.143 & 71.29 & 260.5 & 0\\
scenario1 & cosine & LSQ & estimated & 16.80 & 10.95 & 0.673 & 0.048 & 0.500 & 31.81 & 24.11 & 0\\
scenario1 & cosine & LSQ & fixed\_Ebola\_informed & 273.3 & 243.0 & 0.190 & 0.088 & 0.143 & 50.08 & 0.075 & 0\\
scenario1 & exponential & LAD & estimated & 241.4 & 191.1 & 0.588 & 0.086 & 0.499 & 106.6 & 0.064 & 0\\
scenario1 & exponential & LAD & fixed\_Ebola\_informed & 235.9 & 187.4 & 0.176 & 0.088 & 0.143 & 21.40 & 10.70 & 0\\
scenario1 & exponential & LSQ & estimated & 219.5 & 177.2 & 0.591 & 0.048 & 0.500 & 12.19 & 6.08e-04 & 0\\
scenario1 & exponential & LSQ & fixed\_Ebola\_informed & 183.1 & 144.1 & 0.159 & 0.088 & 0.143 & 8.76 & 0.012 & 0\\
scenario1 & logistic\_decline & LAD & estimated & 73.80 & 47.94 & 0.347 & 0.286 & 0.247 & 0.482 & 0.273 & 0\\
scenario1 & logistic\_decline & LAD & fixed\_Ebola\_informed & 126.1 & 92.25 & 0.478 & 0.088 & 0.143 & 16.71 & 0.030 & 0\\
scenario1 & logistic\_decline & LSQ & estimated & 15.10 & 9.95 & 2.00 & 0.049 & 0.500 & 2.54 & 1.00e-06 & 0\\
scenario1 & logistic\_decline & LSQ & fixed\_Ebola\_informed & 25.90 & 15.60 & 0.196 & 0.088 & 0.143 & 1.00e-06 & 34.94 & 0\\
scenario2 & cosine & LAD & estimated & 103.8 & 69.74 & 0.567 & 0.139 & 0.466 & 1.43 & 4.24 & 0\\
scenario2 & cosine & LAD & fixed\_Ebola\_informed & 25.10 & 15.36 & 0.208 & 0.088 & 0.143 & 6.02 & 3.00 & 0\\
scenario2 & cosine & LSQ & estimated & 20.90 & 14.00 & 0.772 & 0.048 & 0.460 & 4.97 & 1.00e-06 & 0\\
scenario2 & cosine & LSQ & fixed\_Ebola\_informed & 77.80 & 50.94 & 0.177 & 0.088 & 0.143 & 28.13 & 27.18 & 0\\
scenario2 & exponential & LAD & estimated & 111.8 & 82.49 & 0.690 & 0.048 & 0.500 & 91.98 & 0.017 & 0\\
scenario2 & exponential & LAD & fixed\_Ebola\_informed & 74.00 & 46.93 & 0.179 & 0.088 & 0.143 & 12.28 & 74.98 & 0\\
scenario2 & exponential & LSQ & estimated & 350.0 & 307.0 & 0.846 & 0.048 & 0.499 & 0.726 & 0.983 & 0\\
scenario2 & exponential & LSQ & fixed\_Ebola\_informed & 27.40 & 16.52 & 0.165 & 0.088 & 0.143 & 13.17 & 6.59 & 0\\
scenario2 & logistic\_decline & LAD & estimated & 19.00 & 11.70 & 0.607 & 0.100 & 0.258 & 1.13 & 0.014 & 0\\
scenario2 & logistic\_decline & LAD & fixed\_Ebola\_informed & 71.40 & 50.49 & 0.264 & 0.088 & 0.143 & 1.53 & 0.803 & 0\\
scenario2 & logistic\_decline & LSQ & estimated & 56.00 & 36.15 & 0.889 & 0.048 & 0.339 & 4.03 & 0.014 & 0\\
scenario2 & logistic\_decline & LSQ & fixed\_Ebola\_informed & 67.80 & 45.73 & 0.255 & 0.088 & 0.143 & 0.847 & 4.03 & 0\\
scenario3 & cosine & LAD & estimated & 205.8 & 180.2 & 0.152 & 0.048 & 0.048 & 0.878 & 1.40 & 0\\
scenario3 & cosine & LAD & fixed\_Ebola\_informed & 118.0 & 101.0 & 0.360 & 0.088 & 0.143 & 0.012 & 0.002 & 0\\
scenario3 & cosine & LSQ & estimated & 333.0 & 307.0 & 0.345 & 0.141 & 0.300 & 17.96 & 14.53 & 0\\
scenario3 & cosine & LSQ & fixed\_Ebola\_informed & 25.00 & 14.11 & 0.138 & 0.088 & 0.143 & 74.90 & 633.8 & 0\\
scenario3 & exponential & LAD & estimated & 467.1 & 408.3 & 0.742 & 0.048 & 0.496 & 0.247 & 192.7 & 0\\
scenario3 & exponential & LAD & fixed\_Ebola\_informed & 435.9 & 380.8 & 0.185 & 0.088 & 0.143 & 9.66 & 4.83 & 0\\
scenario3 & exponential & LSQ & estimated & 185.7 & 151.6 & 0.174 & 0.048 & 0.170 & 65.56 & 6.55 & 0\\
scenario3 & exponential & LSQ & fixed\_Ebola\_informed & 159.0 & 130.0 & 0.141 & 0.088 & 0.143 & 13.51 & 0.556 & 0\\
scenario3 & logistic\_decline & LAD & estimated & 141.1 & 113.8 & 1.54 & 0.054 & 0.498 & 41.68 & 0.008 & 0\\
scenario3 & logistic\_decline & LAD & fixed\_Ebola\_informed & 55.60 & 41.36 & 0.214 & 0.088 & 0.143 & 22.44 & 0.002 & 0\\
scenario3 & logistic\_decline & LSQ & estimated & 9.40 & 5.71 & 0.616 & 0.235 & 0.500 & 2.78 & 5.45e-05 & 0\\
scenario3 & logistic\_decline & LSQ & fixed\_Ebola\_informed & 18.00 & 11.59 & 0.183 & 0.088 & 0.143 & 74.09 & 1.00e-06 & 0\\
scenario4 & cosine & LAD & estimated & 132.6 & 108.3 & 0.067 & 0.076 & 0.050 & 112.2 & 0.048 & 0\\
scenario4 & cosine & LAD & fixed\_Ebola\_informed & 959.4 & 800.9 & 0.280 & 0.088 & 0.143 & 0.405 & 0.068 & 0\\
scenario4 & cosine & LSQ & estimated & 85.50 & 61.69 & 0.601 & 0.048 & 0.500 & 57.37 & 62.22 & 0\\
scenario4 & cosine & LSQ & fixed\_Ebola\_informed & 102.5 & 59.30 & 0.203 & 0.088 & 0.143 & 15.23 & 0.004 & 0\\
scenario4 & exponential & LAD & estimated & 340.6 & 269.1 & 0.129 & 0.328 & 0.104 & 0.342 & 0.020 & 0\\
scenario4 & exponential & LAD & fixed\_Ebola\_informed & 139.2 & 99.31 & 0.175 & 0.088 & 0.143 & 0.587 & 0.025 & 0\\
scenario4 & exponential & LSQ & estimated & 275.8 & 210.9 & 0.684 & 0.048 & 0.500 & 13.59 & 0.141 & 0\\
scenario4 & exponential & LSQ & fixed\_Ebola\_informed & 286.8 & 222.8 & 0.189 & 0.088 & 0.143 & 3.95 & 2.81 & 0\\
scenario4 & logistic\_decline & LAD & estimated & 124.7 & 75.31 & 0.249 & 0.156 & 0.076 & 0.701 & 1.29e-06 & 0\\
scenario4 & logistic\_decline & LAD & fixed\_Ebola\_informed & 23.90 & 19.53 & 0.629 & 0.088 & 0.143 & 0.096 & 0.084 & 0\\
scenario4 & logistic\_decline & LSQ & estimated & 89.70 & 51.98 & 0.879 & 0.048 & 0.500 & 12.39 & 1.29e-04 & 0\\
scenario4 & logistic\_decline & LSQ & fixed\_Ebola\_informed & 114.8 & 68.69 & 0.547 & 0.088 & 0.143 & 3.89 & 0.001 & 0\\
\end{longtable}
\end{landscape}

\section{Reproducibility specification}
\label{app:reproducibility}
The complete R analysis package is publicly available. It contains the fixed settings, four released scenario files, simulation and fitting functions, execution scripts, task-level results, aggregation code, validation records, and detailed reproducibility instructions. The full analysis used R 4.5.0 and \texttt{deSolve} 1.42 on Linux. Randomness is governed by the deterministic seed scheme below, so every analysis task can be reproduced independently.

\subsection{Parameter bounds}
\begin{longtable}{@{}llll@{}}
\caption{Parameter bounds.}\label{tab:s-bounds}\\
\toprule
Parameter & Lower & Upper & Units / interpretation\\
\midrule
\endfirsthead
\multicolumn{4}{c}{\tablename\ \thetable\ (continued)}\\
\toprule
Parameter & Lower & Upper & Units / interpretation\\
\midrule
\endhead
\bottomrule
\endfoot
beta0 & 0.0001 & 3.0 & per day\\
sigma & 0.0476190476190476 & 0.333333333333333 & per day\\
gamma & 0.0476190476190476 & 0.5 & per day\\
cosine a & -0.95 & 0.95 & dimensionless\\
cosine period & 21.0 & 365.0 & days\\
exponential a & -0.95 & 2.0 & dimensionless\\
exponential b & -0.05 & 0.05 & per day\\
logistic q & 0.05 & 0.95 & long-run fraction\\
logistic k & 0.005 & 0.5 & per day\\
logistic tau & 7.0 & 365.0 & days\\
\end{longtable}

\subsection{Starting values and numerical rules}
\begin{longtable}{@{}p{0.23\textwidth}p{0.70\textwidth}@{}}
\caption{Starting-value generation rules.}\label{tab:s-starts}\\
\toprule
Rule & Description\\
\midrule
\endfirsthead
\multicolumn{2}{c}{\tablename\ \thetable\ (continued)}\\
\toprule
Rule & Description\\
\midrule
\endhead
\bottomrule
\endfoot
Shared starts & LAD and LSQ receive the same feasible starts within every paired simulation comparison.\\
First start & Biologically motivated values clipped to the configured bounds.\\
Additional starts & Independent draws over the parameter bounds; positive parameters are sampled on a log scale.\\
Bootstrap starts & Warm start at the original estimate plus additional feasible draws.\\
Fixed-rate sensitivity & sigma and gamma fixed at Ebola-informed values; all remaining parameters re-estimated.\\
\end{longtable}

\subsection{Seed construction}
\begin{longtable}{@{}p{0.23\textwidth}p{0.70\textwidth}@{}}
\caption{Random-seed scheme.}\label{tab:s-seeds}\\
\toprule
Component & Seed rule\\
\midrule
\endfirsthead
\multicolumn{2}{c}{\tablename\ \thetable\ (continued)}\\
\toprule
Component & Seed rule\\
\midrule
\endhead
\bottomrule
\endfoot
Base seed & 20260627\\
Simulation replication & Deterministic function of base seed, condition ID, and replication ID\\
Rolling-origin task & Deterministic function of base seed, task ID, and forecast origin\\
Fixed-origin bootstrap & Deterministic function of base seed, task ID, bootstrap method, and replication ID\\
Identifiability profiles & Deterministic function of base seed, task ID, profile parameter, and grid index\\
\end{longtable}

\subsection{Negative-binomial dispersion implementation}
The shared helper \texttt{estimate\_nb\_size} first floors model means at $10^{-6}$ and computes
\[
\widehat v=\frac{1}{m}\sum_{i=1}^{m}(y_i^*-\widetilde\mu_i)^2,\qquad
\overline\mu=\frac{1}{m}\sum_{i=1}^{m}\widetilde\mu_i,\qquad
\overline{\mu^2}=\frac{1}{m}\sum_{i=1}^{m}\widetilde\mu_i^2,
\qquad
\widehat\kappa=\frac{\overline{\mu^2}}{\widehat v-\overline\mu}.
\]
No degrees-of-freedom correction is used. If the denominator is nonfinite or at most $10^{-8}$, the fallback size is 30; a finite estimate is truncated to $[0.1,10^6]$. Observed zero counts are retained, and the negative-binomial generator applies a separate mean floor of $10^{-8}$.

For rolling-origin prediction, one size value is \emph{not} held fixed for the whole model--origin fit. In every predictive replication, calibration residuals from the current fit are sampled with replacement to the forecast length; pseudo-counts are formed by adding those sampled residuals to the forecast means, and the formula above is applied to that draw. SEIR models and both baselines use this same routine, but with model-specific residuals: all calibration residuals for SEIR, one-step last-observation residuals for the naive baseline, and residuals from the six-point recent fitting window for the recent-exponential baseline.

For the fixed-origin negative-binomial bootstrap, pseudo-calibration data are first generated using a fit-specific moment estimate based on median-centered calibration residuals. Each bootstrap data set is refitted. The forecast observation layer then uses the same draw-specific residual-resampling calculation described above. Consequently, the reported intervals include refitting variation, residual resampling, and variation in the moment-based size estimate. The implementation is in \texttt{R/utilities.R} (function \texttt{estimate\_nb\_size}), \texttt{R/bootstrap\_methods.R}, and \texttt{R/baseline\_models.R}.

\subsection{Code and data access}
The released RAPIDD scenario data, R package, analysis settings, deterministic random-seed scheme, processed results, tables, figures, automated validation, and detailed reproducibility instructions are publicly available at \url{https://github.com/YisaAdeniyiAbolade/lad-seir-calibration}. 

\subsection{Software}
\begin{longtable}{@{}ll@{}}
\caption{Software versions.}\label{tab:s-software}\\
\toprule
Component & Version\\
\midrule
\endfirsthead
\multicolumn{2}{c}{\tablename\ \thetable\ (continued)}\\
\toprule
Component & Version\\
\midrule
\endhead
\bottomrule
\endfoot
R & 4.5.0\\
deSolve & 1.42\\
platform & x86\_64-pc-linux-gnu\\
\end{longtable}

The ODE solver uses relative and absolute tolerances of $10^{-9}$. Primary LAD/LSQ fits use 12 shared feasible starts, L-BFGS-B, and at most 3,000 iterations per start. Bootstrap refits use a warm start at the original estimate and two additional starts. A fit is retained only when the ODE solution, objective value, state trajectory, and transmission curve are finite and satisfy the configured constraints. Task logs and result-file checks are retained with the computational output.


\newpage
\begin{thebibliography}{99}
\bibitem{ajelli2018rapiddmodel} Ajelli M, Zhang Q, Sun K, Merler S, Fumanelli L, Chowell G, et al. The RAPIDD Ebola Forecasting Challenge: model description and synthetic data generation. \emph{Epidemics}. 2018;22:3--12. doi:10.1016/j.epidem.2017.09.001.

\bibitem{bracher2021wis} Bracher J, Ray EL, Gneiting T, Reich NG. Evaluating epidemic forecasts in an interval format. \emph{PLOS Computational Biology}. 2021;17(2):e1008618. doi:10.1371/journal.pcbi.1008618.

\bibitem{byrd1995lbfgsb} Byrd RH, Lu P, Nocedal J, Zhu C. A limited memory algorithm for bound constrained optimization. \emph{SIAM Journal on Scientific Computing}. 1995;16(5):1190--1208. doi:10.1137/0916069.

\bibitem{cao2011robust} Cao J, Wang L, Xu J. Robust estimation for ordinary differential equation models. \emph{Biometrics}. 2011;67:1305--1313. doi:10.1111/j.1541-0420.2011.01577.x.

\bibitem{chowell2017primer} Chowell G. Fitting dynamic models to epidemic outbreaks with quantified uncertainty: a primer for parameter uncertainty, identifiability, and forecasts. \emph{Infectious Disease Modelling}. 2017;2(3):379--398. doi:10.1016/j.idm.2017.08.001.

\bibitem{chowell2016review} Chowell G, Sattenspiel L, Bansal S, Viboud C. Mathematical models to characterize early epidemic growth: a review. \emph{Physics of Life Reviews}. 2016;18:66--97. doi:10.1016/j.plrev.2016.07.005.

\bibitem{chowell2017perspectives} Chowell G, Viboud C, Simonsen L, Merler S, Vespignani A. Perspectives on model forecasts of the 2014--2015 Ebola epidemic in West Africa: lessons and the way forward. \emph{BMC Medicine}. 2017;15:42. doi:10.1186/s12916-017-0811-y.

\bibitem{chowell2019subepidemic} Chowell G, Tariq A, Hyman JM. A novel sub-epidemic modeling framework for short-term forecasting epidemic waves. \emph{BMC Medicine}. 2019;17:164. doi:10.1186/s12916-019-1406-6.

\bibitem{chretien2015modeling} Chretien JP, Riley S, George DB. Mathematical modeling of the West Africa Ebola epidemic. \emph{eLife}. 2015;4:e09186. doi:10.7554/eLife.09186.

\bibitem{czado2009predictive} Czado C, Gneiting T, Held L. Predictive model assessment for count data. \emph{Biometrics}. 2009;65(4):1254--1261. doi:10.1111/j.1541-0420.2009.01191.x.

\bibitem{efron1993bootstrap} Efron B, Tibshirani RJ. \emph{An Introduction to the Bootstrap}. New York: Chapman \& Hall/CRC; 1993.

\bibitem{gneiting2007proper} Gneiting T, Raftery AE. Strictly proper scoring rules, prediction, and estimation. \emph{Journal of the American Statistical Association}. 2007;102(477):359--378. doi:10.1198/016214506000001437.

\bibitem{hairer1996ode} Hairer E, Wanner G. \emph{Solving Ordinary Differential Equations II: Stiff and Differential-Algebraic Problems}. 2nd ed. Berlin: Springer; 1996.

\bibitem{huber2009robust} Huber PJ, Ronchetti EM. \emph{Robust Statistics}. 2nd ed. Hoboken, NJ: Wiley; 2009.

\bibitem{knight1998limiting} Knight K. Limiting distributions for L1 regression estimators under general conditions. \emph{The Annals of Statistics}. 1998;26(2):755--770. doi:10.1214/aos/1028144858.

\bibitem{koenker2005quantile} Koenker R. \emph{Quantile Regression}. Cambridge: Cambridge University Press; 2005.

\bibitem{kunsch1989jackknife} Kunsch HR. The jackknife and the bootstrap for general stationary observations. \emph{The Annals of Statistics}. 1989;17(3):1217--1241. doi:10.1214/aos/1176347265.

\bibitem{morris2019simulation} Morris TP, White IR, Crowther MJ. Using simulation studies to evaluate statistical methods. \emph{Statistics in Medicine}. 2019;38(11):2074--2102. doi:10.1002/sim.8086.

\bibitem{pell2018phenomenological} Pell B, Kuang Y, Viboud C, Chowell G. Using phenomenological models for forecasting the 2015 Ebola challenge. \emph{Epidemics}. 2018;22:62--70. doi:10.1016/j.epidem.2016.11.002.

\bibitem{qiu2016robust} Qiu Y, Hu T, Liang B, Cui H. Robust estimation of parameters in nonlinear ordinary differential equation models. \emph{Journal of Systems Science and Complexity}. 2016;29:41--60. doi:10.1007/s11424-015-4045-9.

\bibitem{roosa2019identifiability} Roosa K, Chowell G. Assessing parameter identifiability in compartmental dynamic models using a computational approach: application to infectious disease transmission models. \emph{Theoretical Biology and Medical Modelling}. 2019;16:1. doi:10.1186/s12976-018-0097-6.

\bibitem{smirnova2019nonparametric} Smirnova A, deCamp L, Chowell G. Forecasting epidemics through nonparametric estimation of time-dependent transmission rates using the SEIR model. \emph{Bulletin of Mathematical Biology}. 2019;81:4343--4365. doi:10.1007/s11538-017-0284-3.

\bibitem{soetaert2010desolve} Soetaert K, Petzoldt T, Setzer RW. Solving differential equations in R: package deSolve. \emph{Journal of Statistical Software}. 2010;33(9):1--25. doi:10.18637/jss.v033.i09.

\bibitem{viboud2018rapidd} Viboud C, Sun K, Gaffey R, Ajelli M, Fumanelli L, Merler S, et al. The RAPIDD Ebola Forecasting Challenge: synthesis and lessons learnt. \emph{Epidemics}. 2018;22:13--21. doi:10.1016/j.epidem.2017.08.002.

\bibitem{wu1986wild} Wu CFJ. Jackknife, bootstrap and other resampling methods in regression analysis. \emph{The Annals of Statistics}. 1986;14(4):1261--1295. doi:10.1214/aos/1176350142.
\end{thebibliography}
\end{document}